\documentclass[sigconf,nonacm,dvipsnames]{acmart}

\usepackage[a-2b]{pdfx}
\usepackage{ragged2e}
\usepackage{changepage}
\usepackage{url}
\usepackage[normalem]{ulem}
\usepackage{soul}
\usepackage{xcolor}
\usepackage{amsmath}
\usepackage{mathtools}
\usepackage{algorithmicx}
\usepackage[ruled]{algorithm}
\usepackage{algpseudocode}
\usepackage{graphicx}
\usepackage{subcaption}
\usepackage[shortlabels]{enumitem}
\usepackage{cleveref}  
\usepackage{bbm}
\usepackage{stmaryrd}
\usepackage{textcomp}
\usepackage{multicol}
\usepackage{balance}
\usepackage{float}
\usepackage{hhline}
\usepackage{tabularx,booktabs}
\usepackage{multirow}
\usepackage{ulem}

\DeclareMathOperator*{\argmin}{argmin}

\newcommand{\rr}[1]{\textcolor{Black}{#1}}

\definecolor{darkblue}{rgb}{0.0, 0.0, 0.65}
\definecolor{darkred}{rgb}{0.65, 0.0, 0.0}
\definecolor{lightred}{rgb}{0.8, 0.36, 0.36}



\theoremstyle{definition}

\newcommand{\E}[1]{\mathbb{E}\left[#1\right]}
\newcommand{\e}[1]{\mathbb{E}[#1]}

\newcommand{\se}[1]{\mathbb{SE}[#1]}

\newcommand{\var}[1]{\mathbb{V}\mathrm{ar}\left[#1\right]}
\newcommand{\vvar}[1]{\mathbb{V}\mathrm{ar}[#1]}

\renewcommand{\Pr}[1]{{\mathrm{Pr}}\!\left[#1\right]}
\newcommand{\bigO}[1]{\mathcal{O}\left(#1\right)}

\newcommand{\slfrac}[2]{\left.#1\middle/#2\right.}
\newcommand{\w}[1]{|#1|_{\scriptstyle w}}

\makeatletter
\providecommand*{\capdot}{%
	\mathbin{%
		\mathpalette\@capdot{}%
	}%
}
\newcommand*{\@capdot}[2]{%
	\ooalign{%
		$\m@th#1\cap$\cr
		\sbox0{$#1\cap$}%
		\dimen@=\ht0 %
		\sbox0{$\m@th#1\cdot$}%
		\advance\dimen@ by -\ht0 %
		\dimen@=.5\dimen@
		\hidewidth\raise\dimen@\box0\hidewidth
	}%
}

\providecommand*{\cupdot}{%
	\mathbin{%
		\mathpalette\@cupdot{}%
	}%
}
\newcommand*{\@cupdot}[2]{%
	\ooalign{%
		$\m@th#1\cup$\cr
		\sbox0{$#1\cup$}%
		\dimen@=\ht0 %
		\sbox0{$\m@th#1\cdot$}%
		\advance\dimen@ by -\ht0 %
		\dimen@=.5\dimen@
		\hidewidth\raise\dimen@\box0\hidewidth
	}%
}

\providecommand*{\bigcupdot}{%
	\mathop{%
		\vphantom{\bigcup}%
		\mathpalette\@bigcupdot{}%
	}%
}
\newcommand*{\@bigcupdot}[2]{%
	\ooalign{%
		$\m@th#1\bigcup$\cr
		\sbox0{$#1\bigcup$}%
		\dimen@=\ht0 %
		\advance\dimen@ by -\dp0 %
		\sbox0{\scalebox{2}{$\m@th#1\cdot$}}%
		\advance\dimen@ by -\ht0 %
		\dimen@=.5\dimen@
		\hidewidth\raise\dimen@\box0\hidewidth
	}%
}

\makeatother

\title{EdgeSketch: Efficient Analysis of Massive Graph Streams} 

\author{Jakub Lemiesz}
\authornote{This research was funded in whole or in part by National Science Centre, Poland, grant numbers 2023/49/B/ST6/02517 and 2023/07/X/ST6/01598.}
\affiliation{%
  \institution{Wrocław University of S \& T}
  \city{Wrocław}
  \country{Poland}
}
\email{jakub.lemiesz@pwr.edu.pl}

\author{Dingqi Yang}
\authornote{This research was funded in whole or in part by Science and Technology Development Fund, Macau SAR (0011/2025/RIB1, 001/2024/SKL).}
\affiliation{%
  \institution{University of Macau}
  \city{Macau}
  \country{China}
}
\email{dingqiyang@um.edu.mo}

\author{Philippe Cudr\'e-Mauroux}
\affiliation{%
  \institution{University of Fribourg}
  \city{Fribourg}
  \country{Switzerland}
}
\email{philippe.cudre-mauroux@unifr.ch}

\begin{document}
	
\begin{abstract}
	We introduce EdgeSketch, a compact graph representation for efficient analysis of massive graph streams. EdgeSketch provides unbiased estimators for key graph properties with controllable variance and supports implementing graph algorithms on the stored summary directly. It is constructed in a fully streaming manner, requiring a single pass over the edge stream, while offline analysis relies solely on the sketch. We evaluate the proposed approach on two representative applications: community detection via the Louvain method and graph reconstruction through node similarity estimation. Experiments demonstrate substantial memory savings and runtime improvements over both lossless representations and prior sketching approaches, while maintaining reliable accuracy.
\end{abstract}
\maketitle

\section{Introduction}
\label{sec:intro} 	
Across a growing number of domains, from social and communication networks to recommender systems, knowledge graphs, and biological data analysis, connections and relations emerge at such a massive scale that explicit graph representation becomes difficult (see e.g., \cite{UniProt,AdultBrain,BoVWFI,BRSLLP,YahooMusic,TemporalGraph}).
For massive graphs with high average node degrees, say in the order of one hundred, not only adjacency matrices but also adjacency lists require tens or hundreds of gigabytes \mbox{(see e.g. Facebook or host-level web data in \cite{law_datasets}).}
In such cases, storing the graph and executing standard graph algorithms becomes too costly, both in terms of memory and latency. 
This motivates the construction of compact, lossy graph representations that can effectively support a wide range of graph analysis tasks.

A number of papers study \emph{graph embeddings} in the offline setting, where the full graph is available for repeated access.
Methods such as DeepWalk, node2vec, LINE, GraRep, NetMF, and VERSE learn a low-dimensional vector representation for each node, designed to capture different notions of node proximity \cite{DeepWalk,node2vec,LINE,GraRep,NetMF,VERSE}. 
These embeddings support a variety of downstream tasks (e.g., link prediction or node classification), often by serving as input to machine-learning pipelines.
NodeSketch \cite{NodeSketch} fits into this line of work also: it uses recursive sketching to produce embeddings that preserve higher-order proximity, while being faster than prior embedding methods.
More recent work on SGSketch \cite{StreamingGraphEmbeddings} adapts the neighborhood-sketching idea of NodeSketch to the streaming setting with incremental updates, and therefore does not require random access to the full graph.
It processes the edge stream online, updating embeddings for impacted nodes only, and gradually forgets older edges to adapt to structural changes in the graph.
However, all these approaches focus on producing node embeddings rather than providing a universal graph representation that can directly support diverse graph algorithms.

Complementary to embeddings, \emph{graph stream summarization} studies compact data structures that approximate selected statistics of dynamic graphs within a streaming setting. Examples include gSketch, Scube, and Auxo for edge-frequency and heavy-hitter queries \cite{gSketch,Scube,Auxo}, as well as TCM and Graph Stream Sketch for reachability and topological analysis \cite{GSS,FastAndAccurate}. 
By prioritizing low memory usage and fast sketch updates, these methods are highly effective for their intended tasks.
On the other hand, as they are task-specific, they do not provide a universal per-node sketch that can be reused across a wide range of queries or graph algorithms.

In parallel, \emph{data sketches} such as ExpSketch and its more efficient variants, FastExpSketch and FastGM \cite{Lemiesz21,Lemiesz2023,FastGumbelMax}, were proposed as compact and universal summaries for streams of weighted items.
\mbox{Using $O(m)$ memory,} these sketches represent the underlying set $\mathbb{X}$ of distinct stream items and give accurate estimates of its weighted cardinality $|\mathbb{X}|_w = \sum_{x \in \mathbb{X}} w(x)$, with variance decreasing as $O(1/m)$.
Moreover, they support standard operations on sets. Namely, the sketches can be combined to estimate the weight of unions, intersections, differences, and to compute set similarity measures such as the (weighted) Jaccard index.
Although fairly universal, they are designed for unstructured collections and therefore do not capture relationships among items (e.g., temporal order in a stream or connectivity in a graph).

In this work, we bridge these various lines of research and introduce \emph{EdgeSketch}, a compact and universal graph representation built in a fully streaming manner. Conceptually, EdgeSketch inherits from NodeSketch the idea of storing, for each node, a small coordinated sample
of its incident connections, and from FastExpSketch the exponential update rule that enables accurate estimates and supports set operations. Namely, EdgeSketch uses a single pass over the edge stream by treating (possibly weighted) edges as basic stream elements and assigning each node its own sketch. Each node sketch encapsulates both aggregated neighborhood information and a weighted sample of incident edges.
Together, these complementary pieces of information enable a broad spectrum of graph-analytic tasks.
Once the edge stream has been processed, all subsequent analyses can operate entirely on EdgeSketch, without any further access to the original graph.

The rest of the paper is organized as follows.
Section \ref{sec:related} concisely reviews the sketch constructions on which EdgeSketch directly builds, namely ExpSketch, FastExpSketch, and NodeSketch.
\mbox{Section \ref{sec:SES} defines} EdgeSketch and describes its construction. We show that it is time and memory-efficient and naturally handles directed, undirected, and multi-edge graphs, with a straightforward extension to hypergraphs.
Section~\ref{sec:toolbox} develops a toolbox of operations and estimators for EdgeSketch.
We prove that EdgeSketch provides unbiased estimators for key graph parameters (node degrees, edge counts, density) as well as for more sophisticated statistics, all with analytically controlled error. 
In Section \ref{sec:applications}, as an example of a practical application of EdgeSketch, we define a graph modularity estimator and adapt the Louvain method for community detection to operate directly on sketches.
Section~\ref{sec:experiments} empirically evaluates EdgeSketch on graph modularity and graph reconstruction tasks.
The experiments allow us to compare the efficiency of EdgeSketch with NodeSketch and adjacency-list baselines.
\mbox{The presented} results show that EdgeSketch provides a compact per-node representation that supports accurate estimation of key graph statistics and serves as a building block for higher-level \mbox{algorithms.}
\mbox{Together, these} features enable analysis of massive graph streams that would otherwise be difficult or infeasible.

\section{Preliminary}
\label{sec:related} 

In this section, we briefly review the key techniques that inspired our proposed EdgeSketch. We begin with ExpSketch~\cite{Lemiesz21}, a foundational method that enables further enhancements. Next, we present FastExpSketch~\cite{Lemiesz2023}, which significantly accelerates sketch construction and makes the ExpSketch approach scalable to large datasets. We then summarize how FastExpSketch can be used to emulate standard set operations (cf. \cite{Lemiesz2023}). Finally, we recall NodeSketch~\cite{NodeSketch}, which enables the construction of node embeddings in graphs.

\subsection{Exp-Sketch and Fast-Exp-Sketch}
\label{sec:expsketch}
In this section we briefly describe ExpSketch and FastExpSketch (see:  Algorithms~\ref{alg:sketch} and \ref{alg:new}).
Consider a  stream $\mathfrak{M}$ of pairs $(i, \lambda_{i})$, which may repeat, where $i$ is an item identifier and $\lambda_{i} \in \mathbb{R}_{+}$ is its weight. Without loss of generality, for a stream with $n$ distinct items, we assume $i \in \{1, \ldots, n\}$.
The goal is to build a compact sketch using $o(n)$ memory that estimates $\Lambda = \sum_{i=1}^{n} \lambda_i$ with controlled error.

For each stream element $(i, \lambda_{i})$, ExpSketch computes $m$ hashes $h(i\|1), \ldots, h(i\|m)$, where $i\|k$ denotes concatenation. These hashes are treated as independent random variables uniformly distributed on $[0,1]$.
Next, each hash is transformed using the inverse CDF of the exponential distribution: $F^{\,-1}(u) = - \slfrac{\ln u}{\lambda_{i}}$.
By the inverse transform sampling theorem~\cite{Devr86}, each variable $-\ln h(i\|k) / \lambda_{i}$ \mbox{follows} an exponential distribution with parameter $\lambda_{i}$.
Then, we compare point-wise these  $m$ random variables with the values currently stored in sketch $\mathbf{M}$ in line \ref{alg:s} of \mbox{Algorithm \ref{alg:sketch}}.

\begin{algorithm}[t]
	
    \caption{ExpSketch (\,$\mathfrak{M}\,$, $m$ )}
	\label{alg:sketch}
	\begin{algorithmic}[1]
		
		\algrenewcommand\algorithmicrequire{\textbf{Initialization:	}} 
		\Require		
		
		\State	set each of $ m $ positions of sketch $\mathbf{M} = \left(\mathbf{M}_1,\ldots,\mathbf{M}_m\right)$ to $\infty$
		\vspace{0.3cm}	 
		\algrenewcommand\algorithmicrequire{\textbf{Upon 
				arrival of an element $(i, \lambda_{i}) \in \mathfrak{M}$ :}} 
		\Require     
		
		\ForAll{$k\in \{1,2,\ldots, m\}$ }	\label{alg:loop}
		
		\State $U \gets h\left(i \, ||\,  k\right) $ \label{alg:u}
			\State $E \gets  
		{-\slfrac{\ln U}{\lambda_{i}}} $ \label{alg:X}
		
		\State $\mathbf{M}_k \gets \min\left\{\mathbf{M}_k \,,\, E \right\}$ \label{alg:s}
		\EndFor
		
		\vspace{0.3cm}
		\algrenewcommand\algorithmicrequire{\textbf{Upon request, at any time:} }
		\Require     
		\State \textbf{return:} $\mathbf{M}$		
	\end{algorithmic}
\end{algorithm}

Let $E_{i} \sim \mathrm{Exp}(\lambda_i)$ for $i=1,\ldots,n$. After processing all $n$ elements, each sketch cell stores $\mathbf{M}_{k} = \min\{E_{1}, \ldots, E_{n}\}$.
Since the minimum of independent exponential random variables is also exponential, we have $\mathbf{M}_{k} \sim \mathrm{Exp}(\Lambda)$.
The cells $\mathbf{M}_{1}, \ldots, \mathbf{M}_{m}$ are i.i.d., so their sum $G_m = \sum_{k=1}^{m} \mathbf{M}_k$ follows the gamma distribution $G_m \sim \Gamma(m, \Lambda)$.
From the properties of the gamma distribution, it follows that
\begin{equation}
	\label{eq:sum}
	\widebar{\Lambda} = \slfrac{(m-1)}{G_m}
\end{equation}
is an unbiased estimator of $\Lambda$ for $m \geq 2$.
For $m\geq 3$, its variance and standard error are
\begin{equation}
	\label{eq:lambda_variance}
	\vvar{\widebar{\Lambda}} = \frac{\Lambda^2}{m-2}
	\quad \text{and} \quad
	\se{\widebar{\Lambda}} = \frac{1}{\sqrt{m-2}}~.
\end{equation}


	\algrenewcommand\algorithmiccomment[1]{\hfill \textcolor{gray}{$\bigO{#1}$}}

	\begin{algorithm}[t]	
		\caption{FastExpSketch (\,$\mathfrak{M}\,$, $m$ )}
		\label{alg:new}
		\begin{algorithmic}[1]		
			\algrenewcommand\algorithmicrequire{\textbf{Initialization:	}} 
			\Require	
			\State \rr{$permInit \gets (1,2,3,\ldots,m)$} 
			\State	set all values of $\mathbf{M} = \left(\mathbf{M}_1,\ldots,\mathbf{M}_m\right)$ to $\infty$ 
			\State $MAX \gets \infty$ 
			\vspace{0.3cm}	 
			\algrenewcommand\algorithmicrequire{\textbf{Upon 
					arrival of an element $(i, \lambda_{i}) \in \mathfrak{M}$ :}} 
			\Require  
			\State $sum \gets 0 $ 
			\State \rr{$updateMAX \gets false $} 
			\State \rr{$P\gets permInit$} 
			\State $SeedRandom(i)$ 		
			\ForAll{$k\in \{1,2,\ldots, m\}$ }		
			\State $U \gets h\left(i \, ||\,  k\right) $  
			\State $E \gets  
			{-\slfrac{\ln U}{\lambda_{i}}} $
			\State $sum \gets sum + \slfrac{E}{(m-k+1)} $ 
			\State \textbf{if} $sum>MAX$ \textbf{then} break 
			\State \rr{$r \gets RandomInteger([k,m]) $} 
			\State \rr{exchange $P[k]$ and $P[r]$} 
			\State $l \gets P[k] $  
			\State \rr{\textbf{if} $\mathbf{M}_{l}==MAX$ \textbf{then} $updateMAX \gets true $} 
			\State $\mathbf{M}_{l} \gets \min\left\{\mathbf{M}_{l} \,,\, sum\right\}$ 
			\EndFor
			\State \rr{\textbf{if} updateMAX \textbf{then} $MAX \gets \max\left\{\mathbf{M}_1, \ldots, \mathbf{M}_m \right\}$} 
			\vspace{0.3cm}
			\algrenewcommand\algorithmicrequire{\textbf{Upon request, at any time:} }
			\Require     
			\State \textbf{return:} $\mathbf{M}$		
		\end{algorithmic}
	\end{algorithm}

	In ExpSketch, for each stream element we generate $m$ exponential random variables $E_1, \ldots, E_m$ and update $\mathbf{M}_k \leftarrow \min\{\mathbf{M}_k, E_k\}$.
	\mbox{The key} observation is that if $E_k > \max\{\mathbf{M}_1, \ldots, \mathbf{M}_m\}$, then $E_k$ cannot affect the sketch.
	FastExpSketch exploits this by generating order statistics $E_{(1)},  \ldots, E_{(m)}$ of variables $E_1, \ldots, E_m$ incrementally, stopping as soon as $E_{(k)}$ exceeds the current sketch maximum.
	This is enabled by a well-known property of exponential order statistics~\cite{Devr86}: for i.i.d.\ $E_1, \ldots, E_m \sim \mathrm{Exp}(\lambda)$,
	\begin{equation}
		\label{eq:renyi}
		E_{(k)} \stackrel{\mathrm{d}}{=} E_{(k-1)} + \frac{E_k}{m-k+1}~.
	\end{equation}
	Since order statistics lose the original index information, FastExpSketch assigns the generated order statistics to random positions in $\mathbf{M}$.
	This yields a sketch with identical statistical properties to ExpSketch while avoiding unnecessary computation.
	
	The pseudo-code of FastExpSketch is given in Algorithm~\ref{alg:new}. For a detailed description, see~\cite{Lemiesz2023}. We highlight the key points below.
	Variable $sum$ accumulates order statistics according to~(\ref{eq:renyi}). In each iteration, if $sum > MAX = \max\{\mathbf{M}_1, \ldots, \mathbf{M}_m\}$, the loop terminates early since no further values can update the sketch.
	Otherwise, a random position $l$ is selected (lines 13--15) via a lazy Fisher–Yates shuffle~\cite{Knuth97v2}, seeded by the element identifier $i$ (line 7) to ensure consistency across repeated occurrences.
	Finally, the sketch is updated at position $l$, and $MAX$ is recomputed if necessary.

		
\subsection{Set-theory operations}
\label{sec:operations}
A key advantage of ExpSketch and FastExpSketch is their ability to support standard set operations. As shown in~\cite{Lemiesz2023}, sketches of individual sets can be combined to estimate the weighted cardinality of any set $\mathbb{X}$ constructed from sets $\mathbb{A}_1, \mathbb{A}_2, \ldots, \mathbb{A}_d$ using unions, intersections, and complements. The general scheme is as follows.

\begin{enumerate}
	
	\item Create sketches $\mathbf{A}_1,  \ldots, \mathbf{A}_d$ representing sets $\mathbb{A}_1,\ldots, \mathbb{A}_d$.
	
	\item Express set $\mathbb{X}$ with the use of sets $\mathbb{A}_1, \mathbb{A}_2, \ldots, \mathbb{A}_d$  in full disjunctive normal form to obtain an alternative of $b$ disjoint intersections for some $b\geq 1$. 
	
	\item Assume  $i$-th intersection has the form $\mathbb{A}_{1}\ldots\,\mathbb{A}_{r}\,\overline{\mathbb{A}}_{r+1} \, \ldots \, \overline{\mathbb{A}}_{d}$.
	For $i\in\{1,2, \ldots, b\}$, using the Iverson bracket, compute
	$$m'_i =  \sum_{k=1}^{m} \Bigl\llbracket \mathbf{A}_{1,k} = \ldots =  \mathbf{A}_{r,k} < \min\{\mathbf{A}_{r+1,k}, \ldots, \mathbf{A}_{d,k}\} \Bigr\rrbracket~.$$
	
	\item Estimate $\w{\mathbb{X}}$ by summing all unbiased estimates for disjoint intersections. With $m' = m'_1 + \ldots + m'_b$, define
	$$\widehat{\mathbb{X}} = \frac{m'}{m} \cdot \frac{m-1}{\sum_{k=1}^{m}\min\{\mathbf{A}_{1,k}, \ldots, \mathbf{A}_{d,k}\}}~.$$
	Then $\e{\widehat{\mathbb{X}}} = \w{\mathbb{X}}$ and $\vvar{\widehat{\mathbb{X}}} \approx \frac{\w{\mathbb{X}} \cdot \w{\Omega}}{m}$.

\end{enumerate}

\subsection{Node-Sketch}
\label{sec:NodeSketch}

NodeSketch \cite{NodeSketch} generates graph node embeddings through consistent weighted sampling on the node's neighbors. 
For a node $i$ with neighbor set $N(i)$, it produces an embedding of length $m$ by selecting
\begin{equation}
j_l = \argmin_{j \in N(i)\cup\{i\}} \frac{-\log h_l(j)}{w_{ij}}
\end{equation}
for each position $l \in \{1,\dots,m\}$, where $h_l(j) \sim \mathrm{Uniform}(0,1)$ is a position-specific hash seeded by $j$ and $w_{ij} > 0$ is the edge weight.
Note that $-\log h_l(j) / w_{ij}$ follows an exponential distribution with parameter $w_{ij}$, analogously to ExpSketch. Such node embeddings can serve as node features to support various downstream graph analysis tasks.


%

\section{Edge-Sketch}
\label{sec:SES} 

In this section, we present EdgeSketch (see Algorithms \ref{alg:StreamES} and \ref{alg:UpdateES}), which is inspired by NodeSketch and \mbox{FastExpSketch}.
The core concept of EdgeSketch is based on the observation that weights are typically associated with graph edges, making it most convenient to treat edges as basic stream elements to be aggregated in sketches.
Essentially, we consider scenarios where a graph is given as a stream of weighted edges, and storing all the edges becomes infeasible or too costly.
Therefore, the edges are sampled and aggregated in sketches, with each node $i$ having its own associated sketch $\mathbf{M}_i$.
In the following sections, we show that this approach enables significant memory and time savings when analyzing large graphs.
For now, let us assume that edges in stream $\mathfrak{M}$ are directed. In Section~\ref{sec:edges}, we show how EdgeSketch handles other types of edges.

\subsection{Sketch construction}
The steps for constructing a sketch are presented in \mbox{Algorithm \ref{alg:StreamES}.} 
A sketch $\mathbf{M}_i$ for node $i$ is initialized upon the node's first appearance in the stream.
Sketch $\mathbf{M}_i$ consists of two arrays: $\mathbf{F}_i$ storing a weighted sample of edges incident to node $i$, and $\mathbf{S}_i$ containing aggregated information about all incident edges. Initially, $\mathbf{F}_i$ is filled with empty edges denoted as $(0,0)$ and $\mathbf{S}_i$ with infinity values.
The crucial step is the update procedure UpdateES on line 6 that builds upon FastExpSketch. Indeed, by comparing the pseudocodes of FastExpSketch \mbox{(see Algorithm \ref{alg:new})} and UpdateES  (see Algorithm \ref{alg:UpdateES}), one can see that they are in fact very similar. Below we focus on the key modifications and novel elements brought by EdgeSketch.

Since the stream elements are edges, each uniquely identified by its endpoints, EdgeSketch utilizes both endpoints of an edge as a seed for generating random numbers (lines 6 and 8 in Algorithm \ref{alg:UpdateES}). 
To ensure consistent processing of the same edge across all sketches, the order of endpoints is enforced on line 5 of Algorithm \ref{alg:UpdateES}.

A key enhancement of \mbox{EdgeSketch} over \mbox{FastExpSketch} is the introduction of array $\mathbf{F}_i$ for each node $i$, which complements \mbox{array $\mathbf{S}_i$.} Array $\mathbf{F}_i$ maintains a weighted sample of edges incident \mbox{to node $i$} (see lines 15-18 of Algorithm \ref{alg:UpdateES}) and mirrors NodeSketch's approach, where each node stores a weighted sample of its neighbors. 
Observe that the edge \((i,j)\) is inserted into \(\mathbf{F}_{i,l}\) each time the corresponding value \(\mathbf{S}_{i,l}\) is modified.
Thus, according to the FastExpSketch design and from the basic property of the exponential distribution (cf. \cite{Lemiesz2023}), after processing all edges connected to node \(i\), the probability that edge \((i,j)\) is stored in \(\mathbf{F}_{i,l}\) is given by:
\begin{equation}
	\label{eq:fil}
	\Pr{\mathbf{F}_{i,l} = (i,j)} = \slfrac{w_{ij}}{deg_i} \;\;\;\text{where}\;\;\;
	deg_i =\sum\nolimits_{k\in V}w_{ik} 
\end{equation}
denotes the total weight of all edges \mbox{incident to node \(i\).}
The introduction of array \(\mathbf{F}_i\),  though seemingly minor, significantly impacts EdgeSketch functionality.
Arrays \(\mathbf{F}_i\) and \(\mathbf{S}_i\) provide complementary information, each serving distinct purposes. 
Namely, array \(\mathbf{F}_i\) provides information about
 specific connections, while \(\mathbf{S}_i\) offers aggregated information about all connections and their weights.
For instance, the aggregated data in \(\mathbf{S}_i\) can be used to estimate node degrees (see Section~\ref{sec:degrees}), evaluate graph density (see Section~\ref{sec:density}), and, most importantly, facilitate set-theoretic operations on sets of nodes (see Section~\ref{sec:oper}). 
On the other hand,  \(\mathbf{F}_i\) can be used to identify specific neighboring nodes (e.g., for community formation, simulating random walks) and for creating node embeddings as is done in NodeSketch.

The integration of information from array $\mathbf{F}_i$ (similar to the one in NodeSketch) and $\mathbf{S}_i$ (inherited from FastExpSketch) enables the execution of a variety of graph algorithms that are not feasible with NodeSketch or 
FastExpSketch alone. For example, by simulating set-theory operations on $\mathbf{S}_i$, one can create sketches representing sets of vertices (subgraphs).
Then, based on the sample of edges contained in $\mathbf{F}_i$, one can estimate how many edges in a subgraph are internal \mbox{(i.e., both endpoints are in a subgraph)} and how many are external \mbox{(i.e., only one endpoint is in a subgraph)}.
With this crucial feature, EdgeSketch can be used to search for densely connected structures like cliques and communities, and to compute key metrics such as minimal cuts and modularity \mbox{(see Section \ref{sec:applications}).}
In general, taken together, the arrays \(\mathbf{F}_i\) and \(\mathbf{S}_i\)  constitute a comprehensive representation of the graph structure.

\begin{algorithm}[t]	
	\caption{EdgeSketch(\,$\mathfrak{M}\,$, $m$ )} 
	\label{alg:StreamES}
	\begin{algorithmic}[1]
		\algrenewcommand\algorithmicrequire{\textbf{Upon arrival of an edge $(i, j) \in \mathfrak{M}$ with weight $w_{ij}$ :}} 
		\Require
		\algrenewcommand\algorithmicrequire{\textbf{Initialization:	}}
		\Require	 
		\vspace{0.3cm}		 		
		\If{sketch $\mathbf{M}_i = (\mathbf{F}_i,\mathbf{S}_i)$ does not exist}
		\State	set all positions of 1st row   $\mathbf{F}_i = \left(\mathbf{F}_{i,1},\ldots,\mathbf{F}_{i,m}\right)$ to \textcolor{black}{$(0,0)$} 
		\State	set all positions of 2nd row   $\mathbf{S}_i = \left(\mathbf{S}_{i,1},\ldots,\mathbf{S}_{i,m}\right)$ to $\infty$
		\State set  variable  $MAX_i \gets \infty$ 
		\EndIf	
		
		\algrenewcommand\algorithmicrequire{\textbf{Sketch update:	}}
		\Require	
		\vspace{0.3cm}	
		\State  $\mathbf{M}_i , MAX_i =$ UpdateES$(\mathbf{M}_i , MAX_i, \textcolor{black}{(i,j)}, w_{ij})$

		\algrenewcommand\algorithmicrequire{\textbf{Upon request, at any time:} }
		\Require
		\vspace{0.3cm}     
		\State \textbf{return} assemble of all sketches $\mathbf{M} = \left(\mathbf{M}_1, \ldots, \mathbf{M}_n\right)$
				
	\end{algorithmic}
\end{algorithm}

\algrenewcommand\algorithmiccomment[1]{\hfill \textcolor{gray}{//#1}}

\begin{algorithm}[t]	
	\caption{UpdateES ($\mathbf{M}_i$, $MAX_i$, \textcolor{black}{$(i,j)$}, $w_{ij}$ )}
	\label{alg:UpdateES}
	\begin{algorithmic}[1]		
		\algrenewcommand\algorithmicrequire{\textbf{Initialization:	}} 
		\Require
		\State $(\mathbf{F_i},\mathbf{S_i}) \gets \mathbf{M}_i$
		\State $P \gets (1,2,3,\ldots,m)$ 
		\State $sum \gets 0 $ 
		\State $updateMAX \gets false $
		\State \textcolor{black}{$i,j \gets \min(i,j), \max(i,j)$} 		
		\State \textcolor{black}{$SeedRandom(i||j)$}	\Comment{seed for \it{RandomInteger}}
					
		\algrenewcommand\algorithmicrequire{\textbf{Update:	}} 
		\Require
		\vspace{0.3cm}	
		\ForAll{$k\in \{1,2,\ldots, m\}$ }		
		\State \textcolor{black}{$U \gets h\left(i \, ||\, j \, ||\,  k\right) $} 
		\Comment{edge $(i,j)$ and sketch position $k$}
		\State $E \gets  
		{-\slfrac{\ln U}{w_{ij}}} $ \label{alg:calcEEEE}
		\State $sum \gets sum + \slfrac{E}{(m-k+1)} $
		\State \textbf{if} $sum \geq MAX_i$ \textbf{then} break 
		\State $r \gets RandomInteger([k,m]) $
		\State exchange $P[k]$ and $P[r]$
		\State $l \gets P[k] $ 
		\If{$sum  < \mathbf{S}_{i,l} $}
		\State \textbf{if} $\mathbf{S}_{i,l}==MAX_i$ \textbf{then} $updateMAX \gets true $ 
		\State $\mathbf{S}_{i,l} \gets sum$
		\State \textcolor{black}{$\mathbf{F}_{i,l} \gets (i,j)$} \Comment{sample of edges}
		\EndIf
		\EndFor
		\State $ \mathbf{M}_i  \gets (\mathbf{F}_i,\mathbf{S}_i)$
		\State \rr{\textbf{if} updateMAX \textbf{then} $MAX_i \gets \max\left\{\mathbf{S}_{i,1}, \ldots, \mathbf{S}_{i,m} \right\}$} 
		\State \textbf{return} $\mathbf{M}_i, MAX_i$		
	\end{algorithmic}
\end{algorithm}

\subsection{Handling various edge types}
\label{sec:edges}

EdgeSketch can be used for both directed and undirected edges. 
For a directed edge $(v,u)$, we simply use Algorithm \ref{alg:StreamES} to either create or update the sketch for node $v$.	
For an undirected edge, we update the sketches corresponding to both its endpoints.  
Namely, for each undirected edge $\{u,v\}$ from the stream, we should treat it as both $(u,v)$ and $(v,u)$ in the sketching procedure to update sketches of both nodes $u$ and $v$. And since in both sketches we in fact store the same edge, on line 5 of Algorithm \ref{alg:UpdateES} we ensure consistent edge processing in all sketches by ordering its endpoints before adding an edge into a sketch. 
\mbox{In particular,} we note that on line 6 and line 8 of Algorithm \ref{alg:UpdateES} we use both edge endpoints to fix randomness, and their order matters. 
Naturally, we can also add self-loops $(u, u)$ when creating a sketch for node $u$.

The sketch construction enables a straightforward extension to accommodate parallel edges (i.e., multiple edges between the same pair of vertices). To achieve this, one only needs to modify the random number generation on lines~6 and~8 of Algorithm~\ref{alg:UpdateES} so that it takes as input \(i || j || t\) instead of \( i || j \), where \(t\) could be, for instance, a timestamp marking when the edge appears. This proves useful in scenarios where interactions between pairs of vertices occur repeatedly (for example, when multiple transactions take place between the same pair of addresses in a blockchain network). 
Note that the total weight of parallel edges between two nodes can be computed by intersecting their sketches, as this captures exactly the edges incident to both nodes (see Section \ref{sec:oper}).

Finally, it is worth noting that the sketch construction can also be easily extended to support hypergraphs with minimal modifications to the algorithm. Specifically, on line 18 of Algorithm \ref{alg:UpdateES}, we can store not only edges $(u,v)$ but also hyperedges $(u,v,w,\ldots)$. \mbox{This modification} requires a simple adjustment of the pseudorandom number generation on lines 6 and 8 to account for the variable number of endpoints in each hyperedge.

\subsection{Time requirement}
\label{sec:analysis} 

Apart from the initial sketch setup, the time complexity of building EdgeSketch is determined by the update step on line~6 of Algorithm~\ref{alg:StreamES}. Since UpdateES (Algorithm~\ref{alg:UpdateES}) is an adaptation of FastExpSketch, the amortized analysis from \cite{Lemiesz2023} applies: for a node $i$ with \mbox{degree $d_i$}, processing all its incident edges requires $\mathcal{O}(\log d_i)$ expected comparisons. Summing over all nodes, the total expected number of comparisons is
$\mathcal{O}\!\left(\sum_i \log d_i\right)$ and by Jensen's inequality
$$\sum_i \log d_i \le |V| \log \overline{d}, \quad \text{where} \quad \overline{d} = \tfrac{1}{|V|}\sum_i d_i = \tfrac{2|E|}{|V|}$$
is the average degree (for undirected graphs). 
The simple modifications in the pseudo-code compared to FastExpSketch, such as introducing array~$\mathbf{F}_i$, do not affect this bound. Consequently, even for large graphs, the sketch construction remains efficient.

EdgeSketch also supports parallel construction. When the edge stream is partitioned across machines or when memory is constrained, partial sketches can be built independently and merged via position-wise minima on part $\mathbf{S}_i$ with corresponding updates to part $\mathbf{F}_i$ (see Section~\ref{sec:oper}). This merge operation takes $\mathcal{O}(m)$ time per node, independent of the node's degree, and is therefore asymptotically cheaper than building the node sketch itself.

What is worth emphasizing is that, compared to traditional graph representations, using EdgeSketch significantly reduces the time needed to run 
graph algorithms while maintaining controlled accuracy of their results (see Section~\ref{sec:experiments}).  
This speedup results from the purpose-designed, compact sketch representation and the efficiency of set operations on sketches, as outlined in Section~\ref{sec:toolbox}.

\subsection{Memory requirement}
\label{sec:memory}

Regarding memory requirements, each node sketch \(\mathbf{M}_i\) consists of the arrays \(\mathbf{S}_i\) and \(\mathbf{F}_i\) of fixed length $m$. Array \(\mathbf{S}_i\) stores \(m\) floating-point numbers, and array \(\mathbf{F}_i\) stores \(m\) pairs of node identifiers.
We established that node identifiers are integers, and both floating-point numbers and integers use a 64-bit representation. 
\mbox{Therfore, storing} a single sketch \(\mathbf{M}_i\) requires \(3m\) 64-bit numbers. To store the entire graph with $|V|$ nodes, we need \(|V|\) such sketches, which takes \(|V| \times 3m  \) of 64-bit values in total. 
Note that in most applications, the parameter \(m\), which controls the precision of sketch-based estimations (as discussed in Section~\ref{sec:toolbox}), can be set to a \mbox{relatively small value.}

In comparison, representing a weighted graph using an adjacency matrix requires storing \(|V|\times |V|\) of 64-bit values (weights). 
\mbox{Alternatively,} representing the directed graph \(G = (V, E)\) with an adjacency list requires storing  \(|V| +  2 \times  |E|\) of 64-bit values. Specifically, for each node, we maintain a list of its incident edges, and for each edge, we need to store the identifier of the adjacent node and the weight of the edge. For undirected graphs, we typically store each edge twice, hence we need to store \(|V| +  4 \times  |E|\) of 64-bit values.
Therefore, the memory savings provided by EdgeSketch become evident and substantial once the edge set is sufficiently large compared to the number of nodes (cf. Section~\ref{sec:experiments}).
For example, for undirected graphs, comparing memory costs $3m|V| < |V| + 4|E|$ shows that EdgeSketch offers memory savings once $|E| > \tfrac{3m}{4}\,|V|$.

Such scenarios are ubiquitous in modern applications.
One key reason for this is the presence of dense subgraphs, which can significantly boost the total number of edges. For example, this phenomenon often occurs in social networks (e.g., described by stochastic block models \cite{holland1983sbm}), telecommunication systems and sensor networks (e.g., described by geometric graph models \cite{Penrose2003}), and biological networks like brain connectomes, which are known to contain dense \mbox{subgraphs (see e.g. \cite{AdultBrain}).} To illustrate this situation, we note that in the Erdős-Rényi random graph model with \( n \) nodes and edge probability \( p = 1/n \), the expected number of edges is \(\Theta(n ) \).
In contrast, consider a stochastic block model with \( n \) nodes divided into roughly \( n / \log n \) communities, \mbox{each of size \( \log n \).} 
\mbox{We set the probability} of connection between nodes within the same community to \( p = \Theta(1) \), and between nodes in different communities to \( q = 1/n \). 
It is easy to check that the expected number of \mbox{edges in such a random \mbox{graph is \(\Theta(n \log n) \).}}

Overall, EdgeSketch offers significant memory savings over
traditional lossless graph representations as the edge-to-node ratio grows.
Crucially, the proposed solution operates in a fully streaming mode: after processing the stream of edges once, the further offline analysis relies solely on 
the sketch.

\section{Edge-Sketch Toolbox}
\label{sec:toolbox} 
In this section, we present a toolbox of graph analysis tasks that can be performed on EdgeSketch. We begin with basic operations such as estimating node degrees and assessing graph density. We then introduce set-theoretic operations on EdgeSketches that enable more advanced analyses, including finding subgraphs, estimating the number of internal edges, and measuring pairwise node similarity. 
Building on these primitives, we demonstrate that complex graph algorithms, such as community detection, can be implemented directly on the sketch representation without accessing the original graph. 
Throughout this section, we provide theoretical guarantees on the accuracy of the resulting approximations.
For clarity and conciseness, we assume that we are dealing with an undirected graph $G=(V,E)$ in the remainder of the paper.

\subsection{Node degrees}
\label{sec:degrees}
Given a node $i$ in $G$ and its EdgeSketch $\mathbf{M}_i=(\mathbf{F}_i,\mathbf{S}_i)$, we can easily estimate its weighted degree (i.e., the sum of the weights of its incident edges).
The general estimator formula is presented in Equation (\ref{eq:sum}).
For EdgeSketch we get
\begin{equation}
	\label{eq:Wi}
	\widehat{deg}_i = \frac{m-1}{\sum_{j=1}^m \mathbf{S}_{i,j}}~.
\end{equation}
For unweighted graphs, this directly provides an estimate of the node's degree.
From Section \ref{sec:expsketch}  we know  $\widehat{deg}_i $ is an unbiased estimator of $deg_i$
and its variance is given by Equation (\ref{eq:lambda_variance})  for $m\geq 3$:
\begin{equation}
	\label{eq:deg_variance}
	\vvar{\widehat{deg}_i} = \frac{(deg_i)^2}{m-2}~.
\end{equation}
Having unbiased estimators $\widehat{deg}_i$ for all nodes $i\in\{1,2,\ldots,n\}$, we can estimate the degree distribution (e.g., by creating a histogram), 
estimate the average node degree in the graph, and identify the most central nodes (i.e., those with the highest degree).

\subsection{Number of edges and graph density}
\label{sec:density}

Knowing estimates for node degrees, we can estimate the number of edges and the density of a graph.
For example, for an undirected graph $G=(E,V)$ without self-loops, its density $d_G$ is given as the ratio of the number of edges  $|E|$ to the number of all possible edges 
${|V| \choose 2}$.
From the sketch we know the exact value of $|V|=n$, so we just need to \mbox{estimate $|E|$.} We use the fact that for undirected graphs without self-loops 
\begin{equation}
\label{eq:num_edges}
|E| = \frac{1}{2} \, vol(V)
\end{equation}
where, by convention, we define the volume $vol(C)$ of a set of nodes $C$ as the sum of the weighted degrees of the nodes in $C$:
\begin{equation}
	\label{eq:vol_def}
		vol(C) = \sum_{v \in C} \deg_v~. 
\end{equation}
Note that \mbox{$\widehat{vol}(C) = \sum_{v \in C} \widehat{deg}_v$} is an unbiased estimator of $vol(C)$ as it is a sum of unbiased estimators of $deg_v$ defined in formula (\ref{eq:Wi}). 
Finally, we define the unbiased estimator for the graph density as
	$$
	\widehat{d}_G = \frac{\widehat{vol}(V)}{n(n-1)}~.
	$$
Although the estimators $\widehat{\mathit{deg}}(v)$  need not be independent, we can bound $\vvar{\widehat{vol}(C)}$ using the  Minkowski inequality 
\cite{Klenke2014ProbabilityTheory}:
\[
\sqrt{\vvar{\widehat{vol}(C)}}
\le
\sum_{v\in C}\sqrt{\vvar{\widehat{deg}(v)}}.
\]
By Equation (\ref{eq:lambda_variance}) we have
$\vvar{\widehat{\mathit{deg}}(v)}=deg(v)^2/(m-2)$ for $m\ge 3$ and we obtain 
\[
\sqrt{\vvar{\widehat{vol}(C)}}
\le
\sum_{v\in C}\frac{deg(v)}{\sqrt{m-2}}
=
\frac{vol(C)}{\sqrt{m-2}}.
\]
Thus, we have
\[
\vvar{\widehat{vol}(C)}\le \frac{vol(C)^2}{m-2}
\qquad\text{and}\qquad
\se{\widehat{vol}(C)}\le\frac{1}{\sqrt{m-2}}~.
\]

\subsection{Operations on node sketches}
\label{sec:oper}
On node sketches $\mathbf{M}_i=(\mathbf{F}_i,\mathbf{S}_i)$, we can naturally perform set-theory operations described in Section \ref{sec:operations} by using part $\mathbf{S}_i$ of the sketches, which corresponds to FastExpSketch.
It should only be noted that in the case of EdgeSketch, $\mathbf{S}_i$  represents the set of edges with their weights.

Part $\mathbf{F}_i$ of the sketch plays no role in performing set-theory operations on sketches. 
However, the values stored in $\mathbf{F}_i$ may be copied to new sketches along with corresponding values \mbox{from $\mathbf{S}_i$.} 
\mbox{For example,} when creating a new sketch $\mathbf{M}_u=(\mathbf{F}_u,\mathbf{S}_u)$ representing the union of two nodes $i$ and $j$ based on their sketches $\mathbf{M}_i=(\mathbf{F}_i,\mathbf{S}_i)$ and $\mathbf{M}_j=(\mathbf{F}_j,\mathbf{S}_j)$, to obtain the value at position $k$ of $\mathbf{S}_u$ we take the minimum (cf. \cite{Lemiesz2023}): 
\begin{equation}
	\label{eq:Suk}
 \mathbf{S}_{u,k} =\min \left (\mathbf{S}_{i,k},\mathbf{S}_{j,k}\right) 
\end{equation}
and just copy corresponding value to $\mathbf{F}_u$:
\begin{equation}
	\label{eq:Fuk}
	 \mathbf{F}_{u,k} = \begin{cases}
	\mathbf{F}_{i,k} & \text{if } \mathbf{S}_{i,k} < \mathbf{S}_{j,k},\\
	\mathbf{F}_{j,k} & \text{otherwise.}
\end{cases}
\end{equation}
In this way, we get sketch $\mathbf{M}_{u}$, which is exactly the same as we would get by observing the union of edges incident to nodes \mbox{$i$ and $j$.}

As shown in Section \ref{sec:operations} and thoroughly discussed in \cite{Lemiesz2023}, we can evaluate any sequence of set-theory operations using \mbox{sketches $\mathbf{S}_i$.}
Therefore, as we have described above, we can also perform set-theory operations with \mbox{EdgeSketches} $\mathbf{M}_i=(\mathbf{F}_i,\mathbf{S}_i)$.
This capability opens a wide range of applications in graph analysis and graph algorithm design, which we will discuss in subsequent sections.
Here, we present three applications that showcase the usefulness of performing set-theory operations on EdgeSketches:
 
\begin{enumerate}
	\item  \textbf{Efficiently aggregate information about a given node.}
	This feature can be particularly useful when creating graph sketches independently on different machines or when the sketches span different time windows.

	\item \textbf{Merging nodes of a given graph into one super node}. 
	This feature is crucial for finding subgraphs \mbox{(see Section \ref{sec:subgraphs})}, estimating nodes similarity \mbox{(see Section \ref{sec:similarity})} and for \mbox{running} several graph algorithms, like for example \mbox{Karger's Min-Cut.}
	
	\item \textbf{Performing various set-theory operations on graphs.} 
	This feature is essential when we aim to efficiently find overlapping graphs or subgraphs (see Section \ref{sec:similarity}).
\end{enumerate}

\subsection{Finding subgraphs}
\label{sec:subgraphs}
In this section, we show that EdgeSketches may be used not only to estimate the results of set-theory operations but also to decide whether one set is a subset of another.
More specifically,  we demonstrate how to use EdgeSketches to check if a given graph is a subgraph of another graph \mbox{(in the sense that it is a subset of its edges).}

Consider two graphs represented by their sets of \mbox{edges $\mathbb{A}$ and $\mathbb{B}$.}
\mbox{Assume} $\mathbf{M}_A = (\mathbf{F}_A, \mathbf{S}_A) $ and $\mathbf{M}_B = (\mathbf{F}_B, \mathbf{S}_B) $ are the corresponding EdgeSketches and let
 \mbox{$\mathbf{S}_A = (\mathbf{A}_1, \ldots, \mathbf{A}_m)$} and \mbox{$\mathbf{S}_B = (\mathbf{B}_1, \ldots, \mathbf{B}_m)$}.
\mbox{Let} $\mathbb{S_A}$  denote the set of all values  
obtained on line \ref{alg:X} of \mbox{Algorithm \ref{alg:sketch}} for some position $k$ of  sketch  $\mathbf{S}_A$ (and analogously define $\mathbb{S_B}$ for $\mathbf{S}_B$).
 \mbox{Then, for any $k \in \{1,\ldots, m\}$ the sequence of implications holds:}
 $$ \mathbb{A\subseteq B  \Rightarrow S_A\subseteq S_B  \Rightarrow  \min\{S_A\} \geq \min\{S_B\} }  \Rightarrow   \mathbf{B}_k \leq \mathbf{A}_k ~.$$
 As we know that (cf. \cite{Lemiesz21}, Section 4.5)
 $$\Pr{\mathbf{B}_k \leq \mathbf{A}_k} = \Pr{\min\{\mathbb{S_B}\} <  \min\{\mathbb{S_A\setminus S_B}\}   } = \frac{\w{\mathbb{B}}}{\w{\mathbb{A\cup B}}} $$
 we have 
 $$\Pr{(\forall k ) (\mathbf{B}_k \leq \mathbf{A}_k)} = \Pr{\mathbf{B}_k \leq \mathbf{A}_k}^m = \left(\frac{\w{\mathbb{B}}}{\w{\mathbb{A\cup B}}}\right)^m ~.$$
Therefore,
 $$\mathbb{A\subseteq B} \;\Rightarrow\; \Pr{(\forall k ) (\mathbf{B}_k \leq \mathbf{A}_k)} =1  $$
and
  \begin{equation}
\label{eq:costam}
\mathbb{A\nsubseteq B} \; \Rightarrow\;  \Pr{(\exists k ) (\mathbf{B}_k > \mathbf{A}_k)}  = 1- \left(\frac{\w{\mathbb{B}}}{\w{\mathbb{A\cup B}}}\right)^m ~.
  \end{equation}
The conclusion is that with high probability we are able to find a proof that $ \mathbb{A\nsubseteq B}$. A proof is  an index $k$ such that $\mathbf{B}_k > \mathbf{A}_k$.
\mbox{From Equation} (\ref{eq:costam}) we can conclude that the larger part of the weight of set $\mathbb{A}$ is not contained in set $\mathbb{B}$
and the higher the value of parameter $m$, the greater the probability of finding a proof.

\subsection{Subgraph internal edges}
\label{sec:internal}

Until now, we have not discussed the practical applications of the $\mathbf{F}_i$ component of  EdgeSketch $\mathbf{M}_i = (\mathbf{F}_i, \mathbf{S}_i)$, which contains a sample of edges. For the sake of brevity, let us consider an undirected graph $G=(V,E)$.
First and foremost,  $\mathbf{F}_i$ allows us to estimate what proportion of edges is internal to a given subset of nodes. 
We say an edge is internal to a subset of nodes $C\subseteq V$ if both of its endpoints \mbox{belong to $C$.}  
\mbox{The information} about internal edges is essential for analyzing the structural properties of a graph. 
For example, this information is crucial for algorithms that detect communities, find cliques, or compute graph modularity (see Section \ref{sec:modularity}).

The total weight of edges internal to the subset of nodes $C$  
 \mbox{(induced by $C$)} and the total weight of edges incident to nodes in $C$
\begin{equation}
	\label{eq:eC}
	e(C) =  \sum_{\substack{\{v,u\} \in E \\ v, u \in C}} w_{\{v,u\}}	\;\;\;\;\;,\;\;\;\;\;
	w(C) = \sum_{\substack{\{v,u\} \in E \\ v \in C}}  w_{\{v,u\}}
\end{equation}
can be used to compute the weighted proportion of internal edges in the node set $C$ 
\begin{equation}
	\label{eq:p}
	p(C) = \frac{e(C)}{w(C)}~.
\end{equation}
We note that  $w(C)$ differs from $vol(C)$ defined in (\ref{eq:vol_def}) as $w(C)$ counts each edge only once, whereas $vol(C)$ may double-count some edges.

Below, we explain how to estimate $p(C)$, $e(C)$, and $w(C)$.
First, by iteratively computing the union of sketches for the nodes in $C$ based on formula (\ref{eq:Suk}) and (\ref{eq:Fuk}), we create sketch
$\mathbf{M}_C = (\mathbf{F}_C, \mathbf{S}_C)$ representing $C$. 
To obtain the unbiased estimator $\widehat{w}(C)$ of $w(C)$,
we simply use the estimator given in (\ref{eq:Wi}) \mbox{for the sketch $\mathbf{M}_C$}
and get
\begin{equation}
    \label{eq:varwC}
    \var{\widehat{w}(C)} = \frac{w(C)^2}{m-2} \;.
\end{equation}

Note  that array $\mathbf{F}_C$ contains a sample of $m$ edges 
incident to the nodes in $C$, and each edge ${\{u,v\}}$ in this sample is selected independently with probability $w_{\{u,v\}}/w(C)$.
\mbox{This observation can} be justified by viewing sketch $\mathbf{M}_C$ as the representation of a single super node consisting of nodes from $C$ (with possible parallel edges and self-loops) and by using \mbox{Equation (\ref{eq:fil}).}
Now we can define the indicator random variable for the \(k\)th sample edge as
\begin{equation}
	\label{eq:Ik}	
 I_k = \mathbbm{1}\{\text{the \(k\)th sample edge is internal for $C$}\}~.
\end{equation}
Since each \(I_k\) is a Bernoulli random variable with \(\mathbb{E}[I_k] = p(C)\) and \(\operatorname{Var}(I_k) = p(C)(1-p(C))\)~, an unbiased estimator for \(p(C)\) is
\begin{equation}
\label{eq:phat}	
 \hat{p}(C) = \frac{1}{m}\sum_{k=1}^{m} I_k
 \;\;\;\text{and}\;\;\;
 \var{\hat{p}(C)} = \frac{p(C)(1-p(C))}{m}.
\end{equation}

Finally, by formula (\ref{eq:p}), an estimator  for $e(C)$ 
can be obtained as  $$\widehat{e}(C) =  \widehat{w}(C) \cdot \widehat{p}(C)~. $$
Since $\widehat{w}(C)$ and $\widehat{p}(C)$ are both unbiased estimators,
to show that $\widehat{e}(C)$ is also unbiased, we will argue that they are independent.
Namely, note that $\widehat{\mathit{w}}(C)$ depends only on the sketch values in $\mathbf{S}_C$, whereas $\widehat{\mathit{p}}(C)$ depends only on the sampled edges in $\mathbf{F}_C$. By the design of EdgeSketch, different sketch positions are independent.
Moreover, for each position $k$ of the sketch, the identity of the edge selected in $\mathbf{F}_{C,k}$ is independent of the value selected in $\mathbf{S}_{C,k}$.

\begin{lemma}
	\label{lem:race-indep}
	Fix any position $k$ in the sketch, and let $E$ be a set of edges with associated weights $\{w_e\}_{e \in E}$.
	Let $T_e \sim \mathrm{Exp}(w_e)$ be the random variable generated for each edge $e$ on line 9 of Algorithm \ref{alg:UpdateES}.  
	Then,
	$\,S_{C,k} = \min_{e \in E} T_e\,$ and  $\,F_{C,k} = \arg\min_{e \in E} T_e\,$
	are independent random variables. 
\end{lemma}
\begin{proof}
	Let $ W=\sum_{e\in E} w_e$.
	Then based on standard properties of exponential distribution $S_{C,k}\sim \mathrm{Exp}(W)$ and $\mathbb{P}(F_{C,k}=e)=w_e/W$.
	For any $e\in E$ and $t\ge 0$, conditioning on $T_e=s$ gives
	\begin{align*}
		\mathbb{P}(F_{C,k}=e,\ S_{C,k}>t)
		&=\int_t^\infty \mathbb{P}\left(\bigwedge_{e'\neq e} T_{e'} > s\right) f_{T_e}(s)\,ds \\
		&=\int_t^\infty \left(\prod_{e'\neq e} e^{-w_{e'} s}\right)(w_e e^{-w_e s})\,ds \\
		&=\int_t^\infty w_e e^{-Ws}\,ds
		=\frac{w_e}{W}e^{-Wt} \\
		&=\mathbb{P}(F_{C,k}=e)\,\mathbb{P}(S_{C,k}>t).
	\end{align*}
\end{proof}
Lemma~\ref{lem:race-indep} shows that, while edge weights strongly affect which edge appears in $\mathbf{F}_{C,k}$, the sketch value stored in $\mathbf{S}_{C,k}$ does not carry additional information about that identity.

To find the variance of the estimator $\widehat{e}(C)$, recall that for independent random variables $X$ and $Y$
\begin{equation}
\label{eq:var_product}
\var{XY} = \var{X}\var{Y} + \var{X}\E{Y}^2 + \var{Y}\E{X}^2~.
\end{equation}
Thus, with $X=\widehat{\mathit{w}}(C)$,
$Y=\widehat{\mathit{p}}(C)$ and for $m\ge 3$ we obtain
\begin{align}
\var{\widehat{\mathit{e}}(C)}
&= \frac{\mathit{p}\,\mathit{w}^2\,\bigl(m+\mathit{p}-1\bigr)}{m\,(m-2)}
\sim \frac{\,p \,w^2}{m}~,
\label{eq:var_ehat}
\end{align}
for $p = p(C)$ and $w = w(C)$.

\subsection{Random walks}
Another important application of the edge samples stored in $\mathbf{F}_i$ is the simulation of random neighbor selection and random walks, which are fundamental primitives in many graph algorithms.

The $\mathbf{F}_i$ component of EdgeSketch provides, for each node $i$, an explicit weighted sample of edges incident to $i$. In particular, by Equation~\eqref{eq:fil}, an edge $\{i,j\}$ occupies any position $k$ in $\mathbf{F}_i$ with probability $w_{\{i,j\}}/\mathit{deg}(i)$. Consequently, one can simulate a random walk
by repeatedly (and independently) drawing a random sampled incident edge from $\mathbf{F}_i$ at the currently visited node and moving to its other endpoint. Such simulations enable a range of applications, including identifying influential nodes (e.g., PageRank), constructing node $k$-embeddings (e.g., NodeSketch), and detecting densely connected clusters (e.g., the Louvain method; see next section).

\section{Applications}
\label{sec:applications}

In this section, we focus on applications that best showcase the capabilities of EdgeSketch. Using EdgeSketch, we implemented and tested a number of graph algorithms, including the greedy peeling algorithm for densest subgraphs~\cite{Greedy2024}, Karger’s randomized algorithm for minimum cut~\cite{Karger1996}, and a Union--Find approach for identifying giant components~\cite{KNUTH_1978, Thota_2021}. Due to space limitations, we cannot discuss them all. 
Therefore, we decided to present our implementation of the Louvain method for community detection only, as it is somewhat more complex than the other algorithms considered and clearly demonstrates the versatility of EdgeSketch.
However, we would like to emphasize that EdgeSketch makes it straightforward to implement many popular graph algorithms, and to the best of our knowledge, no other sketch or comparably compact representation offers similar capabilities.

In addition, we also report results for the graph reconstruction problem, for two reasons. First, this allows us to demonstrate the improvement over NodeSketch, which is the direct predecessor of EdgeSketch. \mbox{Second, and more importantly,} the reconstruction task is a fundamental and challenging benchmark, testing EdgeSketch's ability to compress and recover the structure of the graph.
\mbox{The source code,} and the technical details of the implementation,
 are available on \href{https://github.com/kubal/EdgeSketch}{\underline{GitHub}}.

\subsection{Community detection}
\label{sec:modularity}
In this section, we show how EdgeSketches can be used to efficiently detect graph communities and estimate graph modularity.
Modularity is a key metric in graph analysis, as it quantifies the strength of community structure. Namely, it indicates how well a graph can be partitioned into densely interconnected subgraphs, offering insight beyond simple triangle-based measures like the clustering coefficient.
Formally, for a weighted undirected graph \( G = (V,E) \) and partition \( \mathcal{P} \) of set $V$, the modularity \( mod_G(\mathcal{P}) \) can be defined as (see e.g. \cite{Fiona_2024,Barabasi2016,ModularityWeighted}):
\begin{equation}
	\!mod_G(\mathcal{P}) = \sum_{C \in \mathcal{P}} mod_G(C) \;,\;\;	
	mod_G(C) =  \frac{e(C)}{e(V)} - \frac{vol(C)^2}{vol(V)^2},
\end{equation}
where $e(C)$ and $vol(C)$ are defined as in formulas (\ref{eq:eC}) and (\ref{eq:vol_def}). 
Then, the modularity of graph \(G\) is defined as the 
maximum modularity value achievable across all possible partitions \( \mathcal{P} \) of the node set \( V \):
$$mod_G = \max_{\mathcal{P}} mod_G(\mathcal{P})~.$$

\noindent
\mbox{One can show} that $mod_G \in [0, 1)$ for $|E|\geq 1$, where 
values close to $1$ indicate dense connections within communities and sparse connections between them, revealing their natural community structure.
For a more detailed discussion and intuition, see \cite{Barabasi2016} and \cite{ModularityWeighted}.

Many community detection algorithms are based on the hypothesis that the partition with the highest modularity corresponds to the optimal community structure.
However, checking all possible partitions is usually computationally infeasible.
A popular approach in this context is the Louvain method (and its derivatives, see e.g.,
\cite{Barabasi2016,newman03fast,Leiden}).
The Louvain method is a greedy algorithm that iteratively optimizes community assignments to maximize modularity. The method has two phases.
In the $1^{st}$ phase, each node is assigned to its own \mbox{community.} 
Then, nodes are repeatedly moved to neighboring communities if such a move increases the modularity. 
Once the modularity can no longer be increased by moving single nodes, the $1^{st}$ phase ends. In the $2^{nd}$ phase, each community is treated as a single node, and the $1^{st}$ phase is repeated on this new graph.

\subsubsection{Louvain method on EdgeSketch}
Here we briefly explain how the Louvain method can be implemented based on EdgeSketch.
In Section \ref{sec:experiments}, we empirically demonstrate that 
this implementation enables significant reductions in time and memory consumption, often by orders of magnitude.
\mbox{This reduction} arises from the compact representation of the graph by the sketch and the efficiency of operations on the sketches.

A crucial step in the Louvain method is to compute the modularity change that results from moving node $v$ \mbox{to community $C$:}
	\begin{equation}
		\label{eq:modGc}
		mod_G(C,v) = mod_G(C\cup{v})  - mod_G(C)~. 	
	\end{equation}
To realize this step with EdgeSketch, first we observe that the movement of node $v$ between communities can be performed along the edges sampled in $\mathbf{F}_v$ component of the sketch $\mathbf{M}_v = (\mathbf{F}_v,\mathbf{S}_v)$.  
	\mbox{Second, it is} easy to check that
	\begin{equation}
		\label{eq:modGc_long}
		 mod_G(C,v)  =  \frac{4\, e(V) \deg_{v|C} \,-\,  2 \, vol(C) \deg_v \,-\, (\deg_v)^2 }{4\, e(V)^2 }~,
	\end{equation}
where $\deg_{v|C} = \sum_{u \in C} w_{\{v,u\}}$	represents the degree of node $v$ within the subgraph induced by community $C$.
Since $e(V)$ is constant for all $C$ and $v$, the denominator $4\, e(V)^2$ can be neglected, and we will show that the nominator can be estimated with sketches. 

Namely, for any $v \in V$ and $C \subseteq V$, the values of $\deg_v$ and $vol(C)$ can be efficiently estimated as shown in \mbox{Sections~\ref{sec:degrees} and~\ref{sec:density}.}
\mbox{The total number} of graph edges $e(V)=|E|$ can be estimated once per algorithm run with formula~(\ref{eq:num_edges}).
To estimate $\deg_{v|C}$, we define indicator random variables analogous to those in Equation (\ref{eq:Ik}). 
\mbox{For the $k$th} edge in the sample $\mathbf{F}_v$, let
\begin{equation*}
	I_k = \mathbbm{1}\{\text{the $k$th sampled edge connects $v$ to a node in $C$}\}
\end{equation*}
Based on formula (\ref{eq:fil}), each edge $\{v,u\}$ in $\mathbf{F}_v$ is sampled with probability $w_{\{v,u\}}/\deg_v$, so we have:
	$$
	\E{I_k} = \sum_{u \in C} \frac{w_{\{v,u\}}}{\deg_v} = \frac{\deg_{v|C}}{\deg_v}~.
	$$
Then, the unbiased estimator $\widehat{deg}_{v|C}$ for $\deg_{v|C}$ can be defined as:
	\begin{equation}
		\label{eq:degvC}
	\widehat{deg}_{v|C} = \widehat{deg}_v \,\cdot\, \frac{1}{m}\sum\nolimits_{k=1}^{m} I_k~.
	\end{equation}
It is unbiased because both estimators on the right are unbiased and they are independent (cf. Lemma~\ref{lem:race-indep}). 
For the same reason, by using  formula (\ref{eq:var_product}) we obtain:
	$\vvar{\widehat{deg}_{v|C}} \sim (deg_{v|C} \cdot  deg_v)/m~.$

To estimate the numerator of \(mod_G(C,v)\) in Equation~(\ref{eq:modGc_long}), we substitute the actual values with their unbiased estimators.
However, the resulting estimate may exhibit a small bias due to variance and covariance terms that arise when taking expectations of products of dependent estimators, such as
	$$\mathbb{E}[(\widehat{deg}_v)^2] =  \mathbb{E}[\widehat{deg}_v]^2 + \mathrm{Var}(\widehat{deg}_v)~,$$ 
		$$\mathbb{E}[\widehat{vol}(C) \cdot \widehat{deg}_v] =  \mathbb{E}[\widehat{vol}(C)] \cdot \mathbb{E}[\widehat{deg}_v] + \mathrm{Cov}(\widehat{vol}(C), \widehat{deg}_v)~.$$
Fortunately, this bias is of order $\bigO{1/m}$, where $m$ is the sketch size, which can be easily shown by using the previously derived variances of the estimators and the Cauchy--Schwarz inequality to bound the covariance terms: $|\mathrm{Cov}(X,Y)| \leq \sqrt{\mathrm{Var}(X)\,\mathrm{Var}(Y)}$.
	
With the ability to efficiently estimate \(mod_G(C,v)\), we can run the Louvain method on EdgeSketch and obtain a graph partition $\mathcal{P}_{ES}$ in the same way as we obtain a partition $\mathcal{P}_{AL}$ when we use the adjacency list,	i.e., by iteratively reassigning nodes \mbox{to communities.} 
\mbox{Note that the partition} \(\mathcal{P}_{ES}\) relies on incomplete knowledge of the graph structure - the sketch offers only a sample of edges that are used to move the nodes between communities, and the estimator of $mod_G(C,v)$ introduces noise.
Consequently, $mod_G(\mathcal{P}_{AL})$ is typically higher than $mod_G(\mathcal{P}_{ES})$, 
with rare exceptions due to randomness.
\mbox{However,} the difference $|mod_G(\mathcal{P}_{AL}) - mod_G(\mathcal{P}_{ES})|$ quickly decreases as the parameter $m$ increases \mbox{(see Figure \ref{fig:dynamic_mod_a} in Section \ref{sec:experiments}).}

\subsubsection{Estimating modularity on EdgeSketch}
\label{sec:mod_estimator}
Computing the exact values of $mod_{G}(\mathcal{P}_{AL})$ and $mod_{G}(\mathcal{P}_{ES})$ requires the full knowledge of graph $G$.
While this is feasible in controlled experiments, it is generally not possible in real large-scale streaming scenarios, where we only have access to the sketch representation of the graph.
\mbox{Therefore, we} define an estimator $\widetilde{mod}_{ES}(\mathcal{P})$ of the modularity $mod_{G}(\mathcal{P})$ based solely on the EdgeSketch. 
Specifically, $mod_{G}(\mathcal{P})$ is computed using the exact values of $e(C)$ and $vol(C)$,  while estimator $\widetilde{mod}_{ES}(\mathcal{P})$ relies on sketch-based estimates:
   \begin{equation}
	\label{eq:mod_estimator}
      \widetilde{mod}_{ES}(\mathcal{P}) = \sum_{C \in \mathcal{P}} \left(
         \frac{\widehat{e}(C)}{\widehat{e}(V)} - \frac{\widehat{vol}(C)^2}{\widehat{vol}(V)^2}
      \right)~.
   \end{equation}

Analogously to the estimator of $mod_G(C,v)$ analyzed above, although estimators $\widehat{e}(C)$ and $\widehat{vol}(C)$ are unbiased, estimator~(\ref{eq:mod_estimator}) involves their ratios and quadratic terms, which may introduce a small bias. As in the case of $mod_G(C,v)$, one could bound the bias using variance and covariance estimates, but the calculations are more involved here. Instead, we use the fact that $\widetilde{mod}_{ES}(\mathcal{P})$ is a smooth function of sample means over $m$ independent sketch positions. 
By the Central Limit Theorem, these sample means are asymptotically normal, and applying the delta method (see e.g., \cite{vanderVaart1998}) to the smooth function $\widetilde{mod}_{ES}(\mathcal{P})$ shows that the bias is $\bigO{1/m}$ and the standard error is $\bigO{1/\sqrt{m}}$. The key insight is that the bias arises from the quadratic term in a second-order Taylor expansion, which scales as $1/m$ due to the $O(1/m)$ variance of the component estimators. The above asymptotics for bias and standard error of $\widetilde{mod}_{ES}(\mathcal{P})$ are empirically validated in Section~\ref{sec:experiments} (see Figure~\ref{fig:dynamic_modularity}b).

\subsubsection{Selection bias}
If we use the same sketch both to \emph{choose} the partition $\mathcal{P}_{ES}$ and to \emph{evaluate} $\widetilde{mod}_{ES}(\mathcal{P}_{ES})$, the reported modularity can be over-optimistic (see: double dipping problem). The cleanest fix is to use two independent sketches with different seeds for randomness. Then, we can compute $\mathcal{P}_{ES}$ on the first sketch and evaluate $\widetilde{mod}_{ES}(\mathcal{P}_{ES})$ on the second.
If only one sketch is available, a simple 2-fold procedure also works: split sketch positions into two halves $A$ and $B$, evaluate $\widetilde{mod}_{ES}(\mathcal{P}_A)$ on $B$ where $\mathcal{P}_A$ is found using $A$, then swap the roles and average the two values.

\subsection{Graph reconstruction}
\label{sec:similarity}

In this section, we show that EdgeSketch can be successfully used to reconstruct a large portion of the edges from the original graph, in addition to the small subset of edges that are explicitly stored in the $\mathbf{F}$ part of the sketch.
The graph reconstruction task is arguably the most challenging and fundamental benchmark, as it directly measures how well the sketch compresses graph information.

A natural approach to edge reconstruction from a sketch is to measure node similarity based on their embeddings.
The intuition is that higher similarity values between nodes indicates a higher probability of an edge existing between them. Using a sketch, we can compute similarities between all node pairs and sort them in descending order to identify the most likely edges.
This approach underlies the solution proposed for NodeSketch and other graph embedding methods mentioned in Section~\ref{sec:intro}. 
\mbox{The effectiveness} of this approach depends critically on the quality of the embeddings and how well they capture the graph structure.

In NodeSketch, embeddings are created by simulating random walks and sampling each node's neighborhood (k-embeddings correspond to samples of their k-neighborhoods). Similarly, for EdgeSketch, we could estimate the similarity between two nodes by simulating random walks based on the $\mathbf{F}$ part of each EdgeSketch and then building and comparing their k-embeddings.
However, with EdgeSketch, it is much more effective to leverage set-theoretic operations, which, as we demonstrate in Section \ref{sec:experiments}, leads to better results compared to NodeSketch in the graph reconstruction task.

Let $\mathcal{E}_d(u)$ denote the set of edges incident to nodes within distance $d$ from $u$.
The weighted Jaccard similarity \cite{Lemiesz2023}, defined as:
$$
	J_d(u, v) =
	 \frac{|\mathcal{E}_d(u) \cap \mathcal{E}_d(v)|_w}{|\mathcal{E}_d(u) \cup \mathcal{E}_d(v)|_w}~,
$$
provides a natural measure of similarity between nodes $u$ and $v$. 
To compute this similarity using EdgeSketch, we first construct $\mathbf{M}_u^d$, the sketch representing $\mathcal{E}_d(u)$, by applying the union operation to sketches of nodes within distance $d$ from node $u$.
Node distances are derived based on edge samples in part $\mathbf{F}$ of each node sketch.
Given sketches $\mathbf{M}_u^d$ and $\mathbf{M}_v^d$, we can efficiently estimate $J_d(u, v)$ by simulating set \mbox{operations (see Section \ref{sec:related} and \cite{Lemiesz2023}).}

Finally, we compute similarity matrices $\mathbf{S}^{(d)} = [J_d(u, v)]_{u,v \in V}$ for different values of $d$ and combine them into a single similarity matrix:
$\mathbf{S} = \sum_{d=0}^{k} \alpha^{\,d} \cdot \mathbf{S}^{(d)}$.
The parameter $\alpha >0$ controls how much weight is given to more distant neighborhoods, with smaller values of $\alpha$ placing more emphasis on local edge structures.

\section{Experiments}
\label{sec:experiments} 
We evaluate EdgeSketch using synthetic and real-world datasets.
\mbox{Experiments} were performed
 on a MacBook Pro with M4 Pro (14-core CPU), 24 GB RAM, and 1 TB SSD.
 For synthetic datasets, we use a single-threaded implementation to ensure consistent runtime comparisons. For real-world datasets, we utilize all 14 CPU threads.
\mbox{The source code for all experiments is available on \href{https://github.com/kubal/EdgeSketch}{\underline{GitHub}}.}

\subsection{Datasets}
\label{sec:datasets}

We conduct experiments on both synthetic graphs generated by the Stochastic Block Model and on large-scale real-world datasets.
\mbox{Using the} Stochastic Block Model allows for controlled experiments, in particular specifying the number of edges and knowing the true modularity, which is unknown for real-world large graphs.

\subsubsection{Stochastic Block Model} 
The Stochastic Block Model (SBM) randomly divides $n$ nodes into $b$ blocks, adding undirected edges between nodes in the same block with probability $p$ and between nodes in different blocks with probability $q$. 
Additionally, each edge $\{u,v\}$ is assigned a random weight:  $w_{\{u,v\}} \sim \mathrm{Exp}(1)$.

For instance, consider graph $\mathbb{G}$ with $n=10^5$ nodes and $b=10$ blocks, \mbox{$p=0.3$} and \mbox{$q=0.03$}, which has approximately $2.85 \times 10^8$ edges.
Storing this graph using an adjacency list requires over \mbox{9.1 GB}, while using an adjacency matrix requires approximately 80 GB. In contrast, EdgeSketch with the size parameter $m=100$ requires only 251 MB of storage and more importantly, as we show below, is sufficient for the considered applications.
One can verify that the memory consumption observed in our experiments aligns with the theoretical analysis presented in Section \ref{sec:memory}.
The memory advantage of EdgeSketch over an adjacency list would become even more significant with an increasing number of edges, as EdgeSketch's size remains constant while the adjacency list grows.

\subsubsection{Epinions} 
\label{sec:Epinions}
We also evaluate EdgeSketch using \texttt{rec-epinions}, a publicly available bipartite graph from the Network Data Repository~\cite{nr}, based on data collected from epinions.com. We will refer to this bipartite graph as $\mathbb{B}$.
It consists of approximately $1.2 \times 10^5$ users and $7.5 \times 10^5$ items, connected by $1.34 \times 10^7$ edges. 
\mbox{Such bipartite} graphs are commonly transformed into unipartite item-item or user-user similarity graphs, which serve as a foundation for collaborative filtering and other recommendation systems (see, e.g., \cite{nr,Hamedani2021TrustRec} and references therein).  
Following this approach, from graph $\mathbb{B}$, we aim to construct an item-item graph, where the weight of an edge between two items is defined as the number of users who rated both items. 
This results in a unipartite weighted \mbox{graph $\mathbb{U}$} with roughly  $7.5 \times 10^5$ nodes and $3.8 \times 10^{10}$ edges. 
Storing this graph as an adjacency list would require \mbox{approximately 1.17~TB.}
Due to memory constraints, constructing graph $\mathbb{U}$ explicitly in a lossless format was infeasible on our machine. 
\mbox{In principle,} one could run graph algorithms on $\mathbb{U}$ without materializing it, just by recomputing edge weights on the fly from~$\mathbb{B}$. However, this approach would be prohibitively slow in practice. Namely, in graph algorithms repeatedly accessing edge weights (e.g., the Louvain method),
each access would require computing an intersection of two neighborhood sets in~$\mathbb{B}$. 
Fortunately, EdgeSketch for graph $\mathbb{U}$ can be effectively built directly from a stream of edges derived \mbox{from graph $\mathbb{B}$,} without materializing $\mathbb{U}$.

Since the graph $\mathbb{U}$ cannot be stored as an adjacency list due to its size, we additionally consider its smaller versions to verify the results of experiments run on EdgeSketch.
Specifically, we filter out edges with weight below a threshold $R$, representing the minimum number of users who rated both items, and remove the resulting isolated vertices. We denote these filtered graphs by $\mathbb{U}_{\geq R}$.
\mbox{For graphs $\mathbb{U}_{\geq R}$ with} sufficiently large $R$, we can store both the EdgeSketch and the adjacency list, enabling direct comparison of the Louvain method on both representations.

\subsection{Modularity estimation}
\label{sec:experiment_mod_est}

We start by presenting sample results of the Louvain method applied to the graph $\mathbb{G}$ described above, illustrating the advantages of EdgeSketch.
Subsequently, we conduct a series of experiments on a smaller SBM graph $\mathbb{H}$ (more details below) to provide a systematic analysis of the performance of the Louvain method on EdgeSketch.

First, let $\mathcal{P}_{T}$ denote the ground-truth partition of the graph nodes. 
Whenever $\mathcal{P}_{T}$ is available for a graph $G$, $mod_G(\mathcal{P}_{T})$ serves as a reference benchmark for modularity estimation. 
For the graph $\mathbb{G}$, we have $mod_{\mathbb{G}}(\mathcal{P}_{T}) = 0.4263$.
\mbox{We applied} the Louvain method to the adjacency list of $\mathbb{G}$ 
a few times, 
\textcolor{black}{each run took about 1h 56m}.
Each time, the method yielded 
a partition $\mathcal{P}_{AL}$ with modularity matching the ground-truth value: \mbox{$mod_{\mathbb{G}}(\mathcal{P}_{AL}) = 0.4263$.}

\mbox{Table \ref{tab:modularity_results} presents} the performance of the Louvain method on EdgeSketch of graph $\mathbb{G}$ with varying sketch size parameters $m$. 
The table shows that increasing the sketch size $m$ improves the accuracy of modularity estimation while maintaining substantial memory and runtime savings.
Note that with $m=100$, EdgeSketch achieves a 36-fold reduction in memory usage and a \textcolor{black}{15-fold} reduction in runtime compared to the adjacency list, 
while providing a reliable modularity estimation.
Let us also note that decreasing $m$ does 
not always lead to a  runtime reduction (see $m=100$, $m=300$). 
This is because smaller values of $m$ can increase the number of steps required for convergence in the Louvain method's first phase.

Second, since running the Louvain method on graph $\mathbb{G}$ is time consuming, making it difficult to conduct a series of systematic experiments, we perform further analyses on a moderately sized SBM graph $\mathbb{H}$ with $n=10^4$, $b=10$, $p=0.5$, $q=0.05$, resulting in approximately $4.7 \times 10^6$ edges with random weights $w_{\{u,v\}} \sim \text{Exp}(1)$.
The EdgeSketch of graph $\mathbb{H}$ requires 25~MB when $m=100$ and 73~MB when $m=300$, compared to 152~MB for the \mbox{adjacency list.}

\begin{table}[t]
    \centering
    \caption{Performance of the Louvain method on EdgeSketch for graph $\mathbb{G}$ across varying sketch size parameters $m$. 
    \mbox{The table shows:} the true modularity of the partition $\mathcal{P}_{ES}$, the estimated modularity $\widetilde{mod}_{ES}(\mathcal{P}_{ES})$, the EdgeSketch storage size in MB, and the Louvain method runtime in minutes.}
    \label{tab:modularity_results}
    \begin{tabular}{|c|c|c|c|c|}
        \hline
        \rule{0pt}{3ex}$m$ & $mod_\mathbb{G}(\mathcal{P}_{ES})$ & $\widetilde{mod}_{ES}(\mathcal{P}_{ES})$ & Size (MB) & \textcolor{black}{Time (m)} \\
        \hline
        50 &  0.2097  & 0.2493   & 131   &  1.3 \\
        \hline
        100 & 0.4171   & 0.4335   &  251  & 7.9  \\
        \hline
        300 &  0.4263  & 0.4324   & 731  & 4.8  \\ 
		\hline
        500 &  0.4263  & 0.4395   & 1211  & 8.8  \\
		\hline
        1000 & 0.4263  & 0.4264   & 2411  & 17.0 \\
        \hline
    \end{tabular}    
\end{table}

\begin{figure*}[!t]
    \centering
    \begin{subfigure}[t]{0.48\textwidth}
        \centering
        \includegraphics[width=\linewidth]{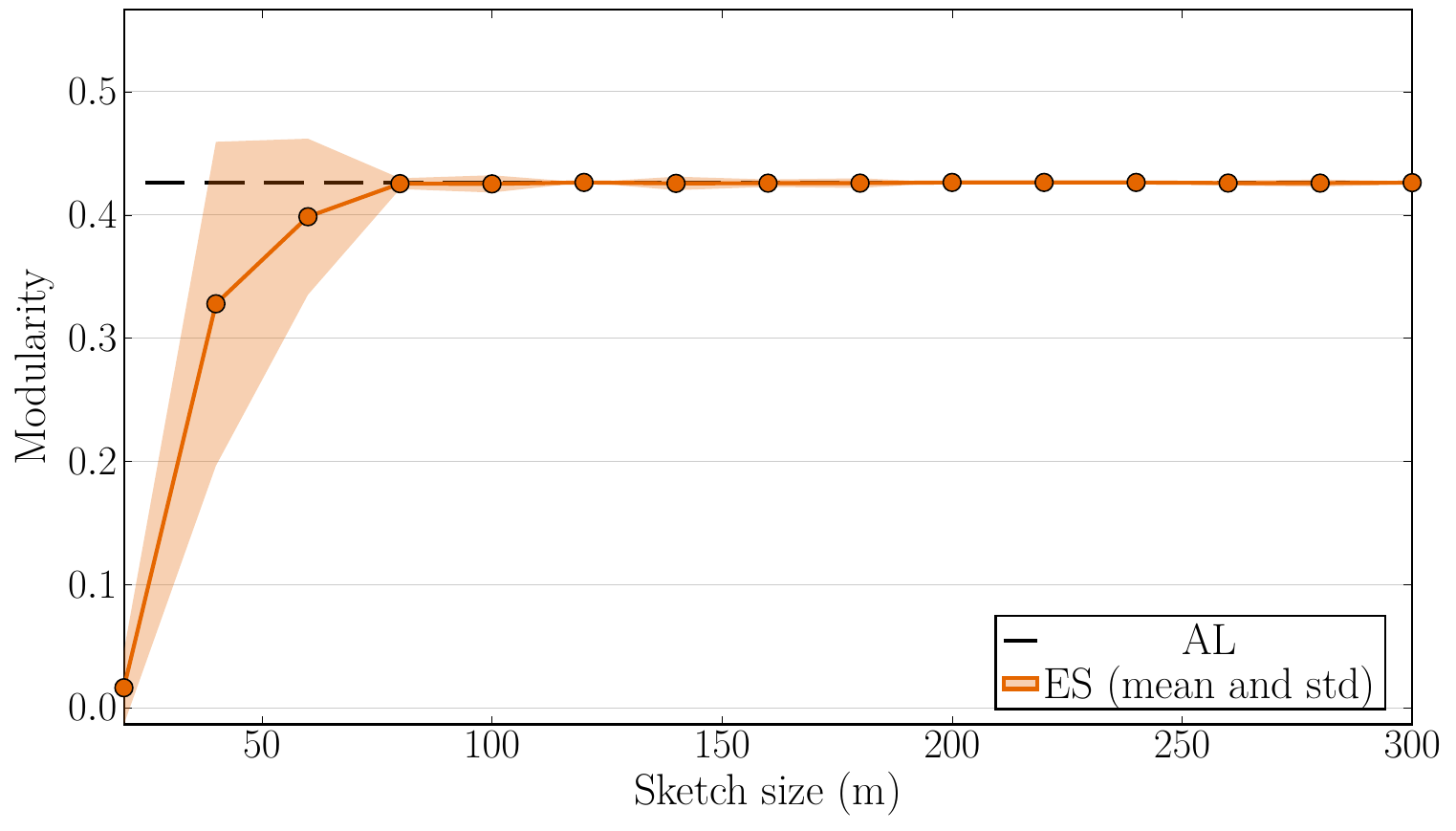}
        \caption{Comparison of $mod_{\mathbb{H}}(\mathcal{P}_{AL})$ and $mod_{\mathbb{H}}(\mathcal{P}_{ES})$.}
        \label{fig:dynamic_mod_a}
    \end{subfigure}%
    \hfill
    \begin{subfigure}[t]{0.48\textwidth}
        \centering
        \includegraphics[width=\linewidth]{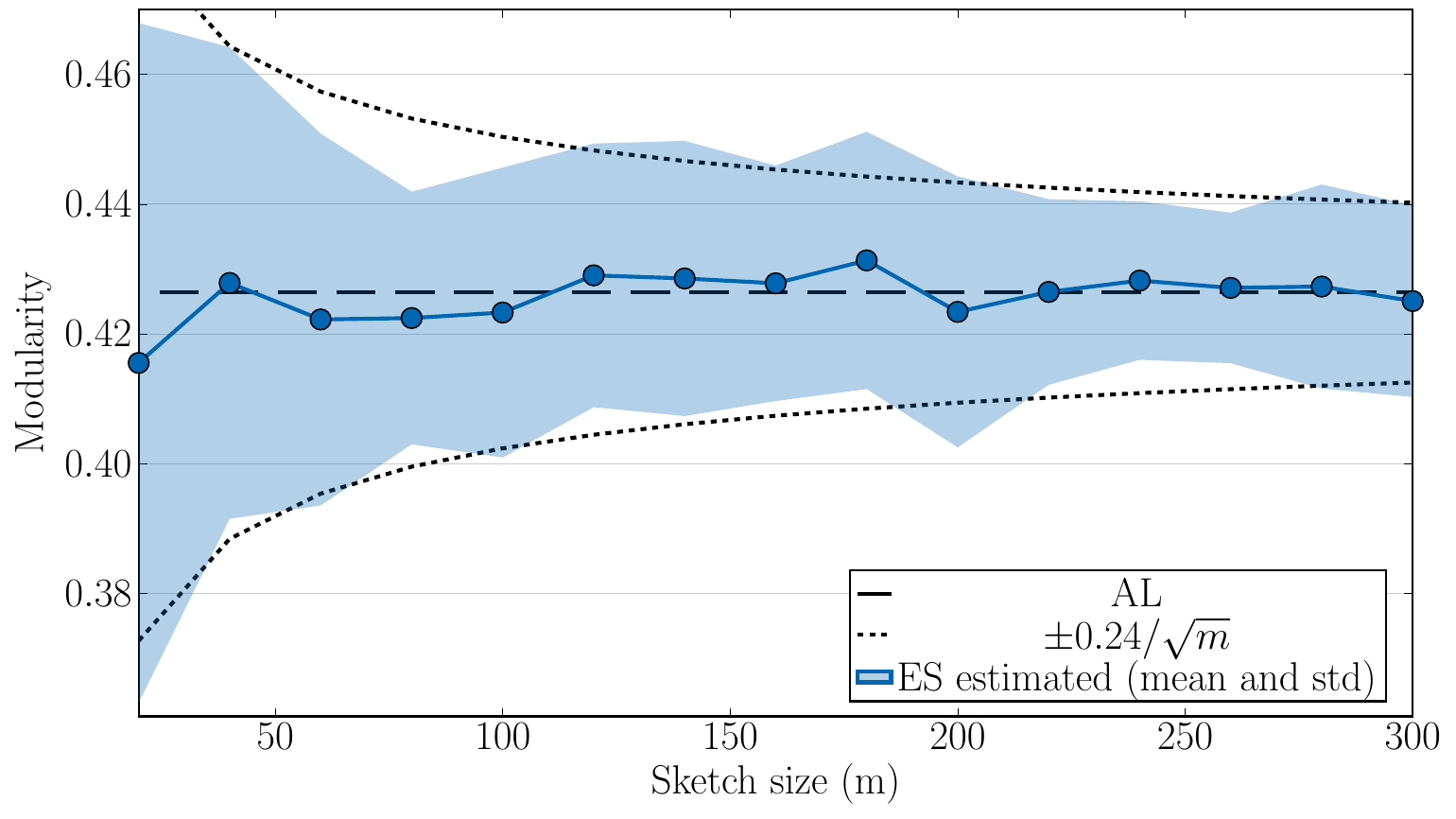}
        \caption{Comparison of $mod_{\mathbb{H}}(\mathcal{P}_{AL})$ and $\widetilde{mod}_{ES}(\mathcal{P}_{AL})$.}
        \label{fig:dynamic_mod_b}
    \end{subfigure}

    \vspace{0.5em}

    \begin{subfigure}[t]{0.48\textwidth}
        \centering
        \includegraphics[width=\linewidth]{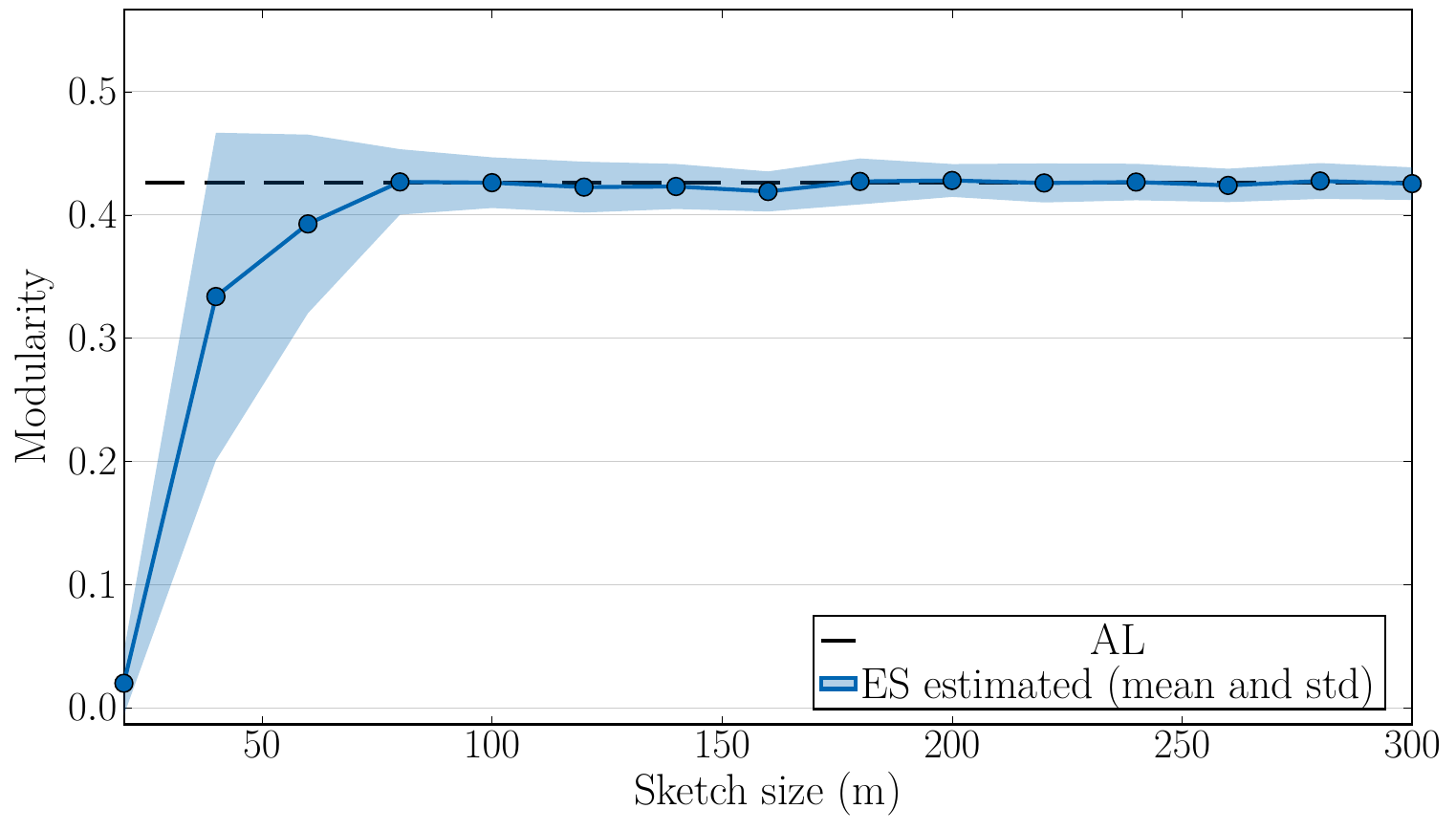}
        \caption{Comparison of $mod_{\mathbb{H}}(\mathcal{P}_{AL})$ and $\widetilde{mod}_{ES}(\mathcal{P}_{ES})$.}
        \label{fig:dynamic_mod_c}
    \end{subfigure}%
    \hfill
    \begin{subfigure}[t]{0.48\textwidth}
        \centering
        \includegraphics[width=\linewidth]{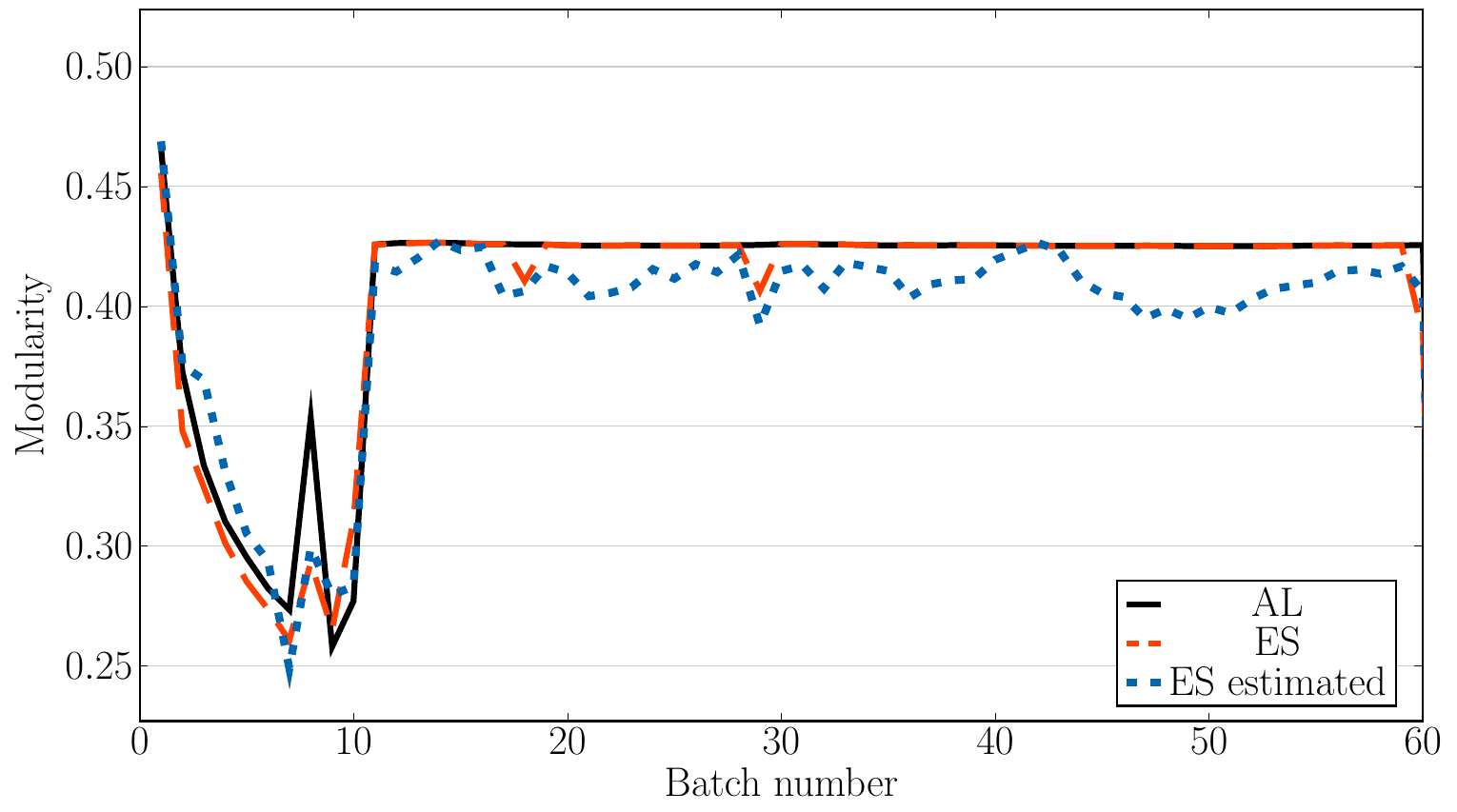}
        \caption{Comparison of $mod_{\mathbb{H}}(\mathcal{P}_{AL})$, $mod_{\mathbb{H}}(\mathcal{P}_{ES})$, and $\widetilde{mod}_{ES}(\mathcal{P}_{ES})$ in a dynamic graph setting (sketch size parameter $m=100$).}
        \label{fig:dynamic_mod_d}
    \end{subfigure}
    \vspace{-0.1cm}
    \caption{Experimental results on graph $\mathbb{H}$. 
    In Figures (a)--(c), the reference value $mod_{\mathbb{H}}(\mathcal{P}_{AL})$ is denoted by a dashed line, while EdgeSketch-based estimates are computed for sketch sizes $m$ ranging from 20 to 300 in steps of 20. For each $m$, we run 50 trials, reporting the mean (solid line) and standard deviation (shaded region). 
        Figure~(d) depicts the dynamic setting: edges arrive in 60 batches of $6 \times 10^4$ randomly selected edges each, with three modularity approximations computed after every \mbox{batch ($m=100$).}
        }
    \label{fig:dynamic_modularity}
\end{figure*}

Figure~\ref{fig:dynamic_modularity} presents experimental results comparing different modularity estimations on graph $\mathbb{H}$. 
Figure \ref{fig:dynamic_mod_a} compares $mod_{\mathbb{H}}(\mathcal{P}_{AL})$ and $mod_{\mathbb{H}}(\mathcal{P}_{ES})$, where both partitions $\mathcal{P}_{AL}$ and $\mathcal{P}_{ES}$ are evaluated on the original graph $\mathbb{H}$. 
For $mod_{\mathbb{H}}(\mathcal{P}_{ES})$, we change the sketch size $m$ from 20 to 300 in increments of 20. 
The figure shows that as $m$ increases, the bias and standard deviation of $mod_{\mathbb{H}}(\mathcal{P}_{ES})$ decrease, and for $m\geq 100$, they become nearly negligible.

Figure \ref{fig:dynamic_mod_b} compares $mod_{\mathbb{H}}(\mathcal{P}_{AL})$ with sketch-based estimates $\widetilde{mod}_{ES}(\mathcal{P}_{AL})$ for a fixed partition $\mathcal{P}_{AL}$ for various sketch sizes $m$.
The results confirm that $\widetilde{mod}_{ES}$ is asymptotically unbiased with standard deviation scaling as $1/\sqrt{m}$, as discussed in Section~\ref{sec:mod_estimator}. The theoretical bounds (with empirically fitted constant $c=0.24$) closely match the observed behavior of $\widetilde{mod}_{ES}$.

Figure \ref{fig:dynamic_mod_c} compares $mod_{\mathbb{H}}(\mathcal{P}_{AL})$ with $\widetilde{mod}_{ES}(\mathcal{P}_{ES})$, representing a fully sketch-based approach where both the partition $\mathcal{P}_{ES}$ and the modularity estimate $\widetilde{mod}_{ES}$ are derived entirely from EdgeSketch.
The figure shows that as $m$ increases, the bias and standard deviation of $\widetilde{mod}_{ES}(\mathcal{P}_{ES})$ decrease. As we combine two sources of uncertainty, the bias exhibits behavior analogous to that of $mod_{\mathbb{H}}(\mathcal{P}_{ES})$ in Figure~\ref{fig:dynamic_mod_a}, while the standard deviation follows the pattern observed for $\widetilde{mod}_{ES}(\mathcal{P}_{AL})$ in Figure~\ref{fig:dynamic_mod_b}. 

Figure \ref{fig:dynamic_mod_d} illustrates modularity estimation in a dynamic setting. 
We randomly sample batches of $5 \times 10^4$ edges from the full set of approximately $4.7 \times 10^6$ edges and sequentially add them to an initially empty graph.
After each batch, we compute $mod_\mathbb{H}(\mathcal{P}_{AL})$, $mod_\mathbb{H}(\mathcal{P}_{ES})$, and $\widetilde{mod}_{ES}(\mathcal{P}_{ES})$.
Initially, all Louvain-based modularity values deviate from the ground-truth $mod_\mathbb{H}(\mathcal{P}_{T}) = 0.4265$, as the block structure is harder to identify with fewer edges.
\mbox{However,} apart from minor fluctuations, all three estimates 
remain close to each other,
 and from approximately the $12^{th}$ batch onward, they align with $mod_\mathbb{H}(\mathcal{P}_{T})$. 
In this experiment, we use EdgeSketch with $m=100$, providing a 6-fold reduction in memory (25~MB vs.\ 152~MB for the adjacency list).

\subsection{Graph reconstruction} 

Below, we present experimental results for the graph reconstruction problem described in Section \ref{sec:similarity}.
We compare EdgeSketch to NodeSketch, which has been extensively benchmarked against other edge prediction methods in~\cite{NodeSketch,StreamingGraphEmbeddings}.
To ensure a fair comparison, we normalize the memory budget across methods. 
Each EdgeSketch element requires $64$ bits for part $\mathbf{S}$ and $2\times 64$ bits for part $\mathbf{F}$, while NodeSketch stores only a 64-bit node identifier per element.
\mbox{Since the sampled} edges in $\mathbf{F}$ directly contribute to reconstruction, we evaluate EdgeSketch both with and without storing part $\mathbf{F}$, denoting the latter version as $\text{EdgeSketch}^*$.
Thus, with a fixed memory budget, if EdgeSketch uses $m$ elements per node, then NodeSketch and $\text{EdgeSketch}^*$ are each allocated $3m$ \mbox{elements per node.}

We build all three sketches to obtain $k$-embeddings of graph nodes.
The values of $k=4$ and $\alpha=0.2$ were tuned experimentally to maximize precision  
for \mbox{graph $\mathbb{H}$.}
For EdgeSketch and $\text{EdgeSketch}^*$, we use part $\mathbf{S}$ to find weighted Jaccard similarities $J_d(u, v)$ across different distances $d$ and combine them into a single similarity score using the weighted sum formula  (see Section~\ref{sec:similarity}).
Edges stored in part $\mathbf{F}$ of EdgeSketch are assigned 100\% similarity.
\mbox{For NodeSketch,} $k$-embeddings are used to find node similarity via weighted Jaccard similarity based on matching positions (see~\cite{NodeSketch}).

For all sketches, we select the top $t$ node pairs with the highest similarity scores as the predicted edges. 
Let $\widehat{E}_t$ denote the set of $t$ predicted edges and $E$ the set of true edges. The precision of the prediction is then defined as $P_t = |\widehat{E}_t \cap E| / t$.
Figure~\ref{fig:reconstruction} shows the edge reconstruction precision $P_t$ for graph $\mathbb{H}$ as $t$ ranges from 0\% to 100\% of the total number of edges $|E|$.
Note that the baseline precision for random edge guessing is approximately $9.5\%$ (the graph $\mathbb{H}$ density), whereas knowledge that two nodes are in the same cluster yields a baseline precision of 50\% ($p=0.5$ in SBM).

In Figure~\ref{fig:reconstruction}a, we fix $m=100$, which results in about 25 MB of memory for each of the three sketches.
For EdgeSketch, this gives \mbox{$n \cdot m = 10^6$} sampled edges in part $\mathbf{F}$, potentially covering roughly 21\% of edges in \mbox{graph $\mathbb{H}$.} 
In practice, the number of unique sampled edges in  $\mathbf{F}$ is smaller: the same edge may be sampled many times and can appear in sketches of both endpoints. 
As a result, unique sampled edges account for roughly 15\% of the total edge count, and thus for \mbox{$t \leq 15\%  |E|$} EdgeSketch achieves 100\% precision.
\begin{figure}[]
    \centering
    \begin{subfigure}[t]{0.48\textwidth}
        \centering
        \includegraphics[width=\linewidth]{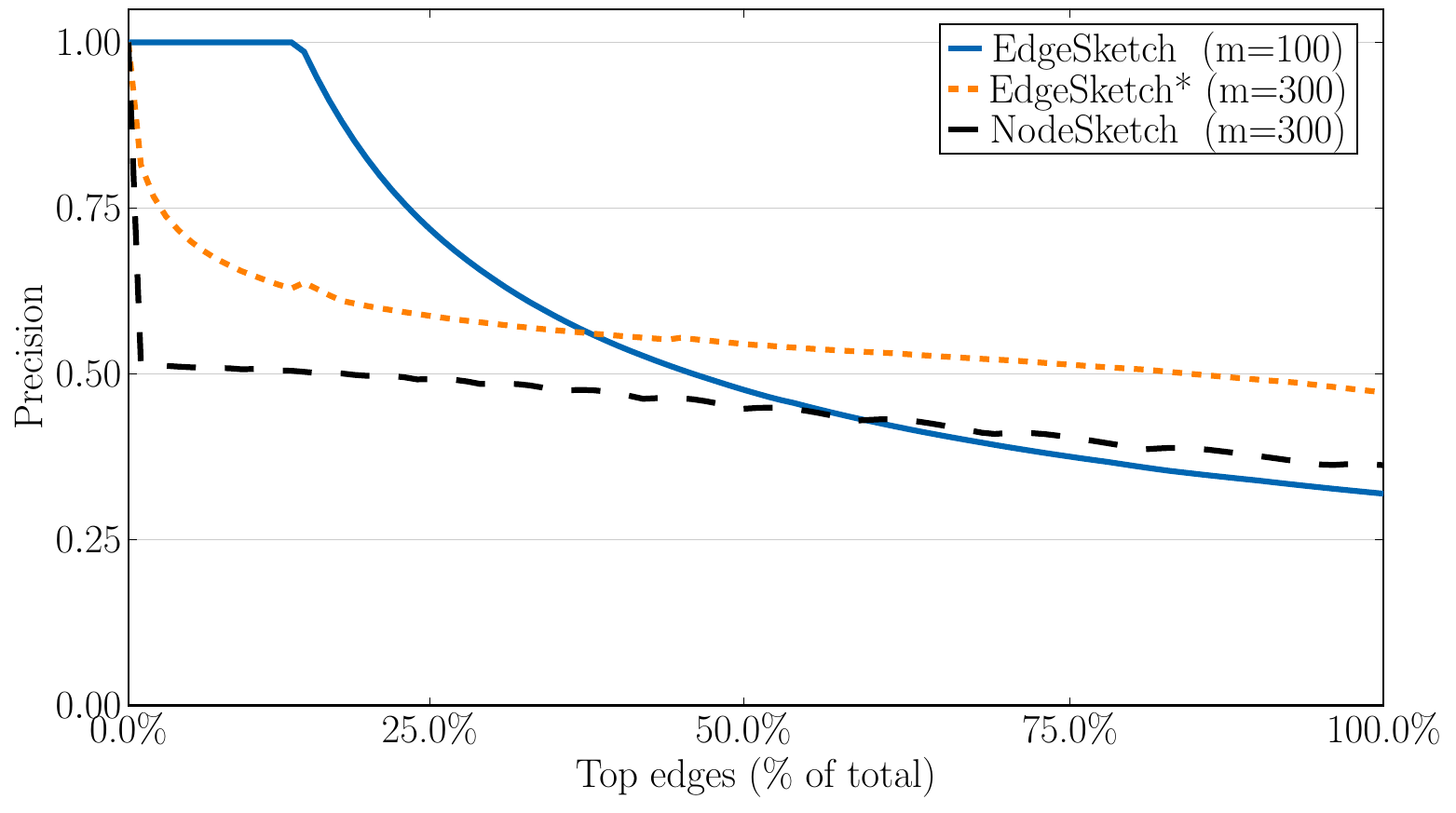}
        \caption{Memory budget of approximately 25~MB per sketch.}
        \label{fig:reconstruction_m100}
    \end{subfigure}%
    \hfill
    \begin{subfigure}[t]{0.48\textwidth}
        \centering
        \includegraphics[width=\linewidth]{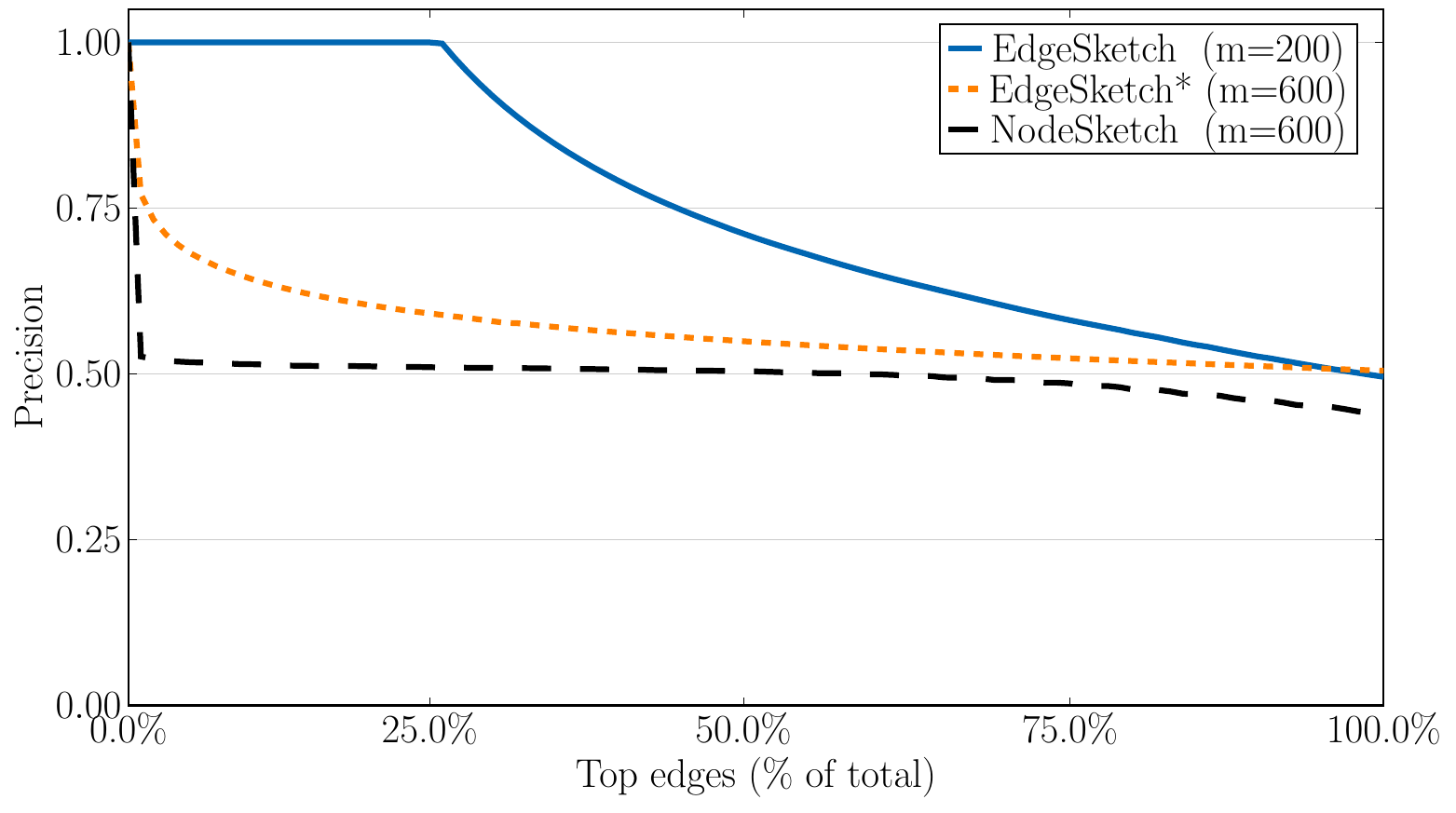}
        \caption{Memory budget of approximately 50~MB per sketch.}
        \label{fig:reconstruction_m200}
    \end{subfigure}
    \caption{Edge reconstruction precision $P_t$ on graph $\mathbb{H}$ for three sketches under equal memory budgets, with embedding depth $k=4$ and decay parameter $\alpha=0.2$.}
    \label{fig:reconstruction}
\end{figure}
For larger $t$, EdgeSketch maintains relatively high precision, exceeding both NodeSketch and $\text{EdgeSketch}^*$, before eventually falling below NodeSketch. 
\mbox{NodeSketch} gradually declines from 50\% precision to roughly 35\%. 
For all $t$, $\text{EdgeSketch}^*$ \mbox{outperforms NodeSketch.} 

In Figure~\ref{fig:reconstruction_m200} we double the memory budget to approximately 50~MB per sketch by setting $m=200$. This allows EdgeSketch to sample more edges, extending the region of perfect precision to \mbox{$t  \leq 28\% |E|$}.
Note that as $t$ increases, the precision of all three sketches converges to approximately 50\% or below, regardless of the sketch size, matching the within-cluster edge probability $p=0.5$.
This suggests that all three sketches can infer whether two nodes belong to the same cluster but struggle to distinguish which pairs within a cluster are connected.
However, for  moderate values {of $t$,} which are most relevant in practical applications, EdgeSketch achieves significantly higher precision than the \mbox{other sketches.}


\newpage
\subsection{Epinions graph}
For the graph $\mathbb{U}$ described in Section~\ref{sec:Epinions}, we construct an EdgeSketch with parameter $m=100$. 
The construction takes roughly 6.4h and the resulting sketch occupies 1.81~GB of memory, only about 0.15\% of the adjacency list size.
We then run the Louvain algorithm on this EdgeSketch, which takes 
1h 32min
and returns an estimated modularity of 0.257. The resulting partition contains 283 clusters. 

Since the full graph $\mathbb{U}$ cannot be stored as an adjacency list, we cannot verify this partition against a baseline.
To evaluate the accuracy of EdgeSketch-based partitioning, we additionally consider smaller graphs $\mathbb{U}_{\geq 5}$, $\mathbb{U}_{\geq 10}$, $\mathbb{U}_{\geq 20}$, $\mathbb{U}_{\geq 30}$ (see Section~\ref{sec:Epinions}).
\mbox{For each of these graphs, we construct EdgeSketch with $m=100$}.

Table~\ref{tab:graphs} and Table \ref{tab:louvain} present graph characteristics, EdgeSketch sizes, 
and results for the Louvain method. 
The results confirm the high efficiency of EdgeSketch compared to the adjacency list representation. 
In particular, we observe substantial memory savings \mbox{(up to \mbox{646-fold})} and runtime speedup (up to 34-fold) for the Louvain method, while maintaining high estimation accuracy.

\begin{table}[b]
    \centering
    \caption{EdgeSketch (ES, $m=100$)  construction time in minutes and memory gain over adjacency list (AL).} 
    \label{tab:graphs}
    \begin{tabular}{|c|c|c|r|r|r|r|}
        \hline
        \multirow{2}{*}{Graph} & \multirow{2}{*}{$|V|\approx$} & \multirow{2}{*}{$|E|\approx$} & \multicolumn{3}{c|}{Memory (GB)} & \multirow{2}{*}{\textcolor{black}{Time (m)}} \\
        \cline{4-6}
              &             &             & \multicolumn{1}{c|}{AL} & \multicolumn{1}{c|}{ES} & \multicolumn{1}{c|}{Gain} & \\
        \hline
        $\mathbb{U}_{\geq 30}$ & 69k & $5 \times 10^7$ & 1.8 & 0.17 & 10$\times$ & 0.5 \\
        \hline
        $\mathbb{U}_{\geq 20}$ & 107k & $2 \times 10^8$ & 5.0 & 0.27 & 19$\times$ & 1 \\
        \hline
        $\mathbb{U}_{\geq 10}$ & 208k & $6 \times 10^8$ & 19.8 & 0.52 & 38$\times$ & 5 \\
        \hline
        $\mathbb{U}_{\geq 5}$ & 448k & $2 \times 10^9$ & 60.8 & 1.13 & 54$\times$ & 17 \\
        \hline
        $\mathbb{U}$ & 755k & $4 \times 10^{10}$ & 1\,169 & 1.81 & 646$\times$ & 383 \\
        \hline
    \end{tabular}
\end{table}

\begin{table}[b]
    \centering
    \caption{Louvain method: modularity and runtime comparison between EdgeSketch (ES) and adjacency list (AL).}
    \label{tab:louvain}
    \begin{tabular}{|c|r|r|r|c|c|c|}
        \hline
        Graph & \multicolumn{3}{c|}{\textcolor{black}{Time (m)}} & \multicolumn{3}{c|}{Modularity} \\
        \cline{2-4} \cline{5-7}
              & \multicolumn{1}{c|}{AL} & \multicolumn{1}{c|}{ES} & \multicolumn{1}{c|}{Gain} & AL & ES & est. ES \\
        \hline
        $\mathbb{U}_{\geq 30}$ & 3.1 & 0.5 & 6$\times$ & 0.476 & 0.481 & 0.455 \\
        \hline
        $\mathbb{U}_{\geq 20}$ & 14.4 & 1.6 & 9$\times$ & 0.378 & 0.371 & 0.427 \\
        \hline
        $\mathbb{U}_{\geq 10}$ & 112 & 3.0 & 34$\times$ & 0.324 & 0.317 & 0.324 \\
        \hline
        $\mathbb{U}_{\geq 5}$ & 532 & 27 & 20$\times$ & 0.279 & 0.270 & 0.251 \\
        \hline
        $\mathbb{U}$ & --- & 92 & --- & --- & --- & 0.257 \\
        \hline
    \end{tabular}
\end{table}

For the EdgeSketch of graph $\mathbb{U}_{\geq 30}$, we also perform edge reconstruction.
Computing the full similarity matrix would be very time- and memory-consuming, so we restrict the computation to pairs of nodes within k hops of each other, where distances are determined using edges in the F component of the sketch. Since similarity between nodes is computed by comparing their sketch-based k-neighborhoods, considering pairs at greater distances becomes ineffective as the overlap between their neighborhoods decreases rapidly, leading to near-zero similarity scores.
For  $k=2$, $\alpha=1$ and $t = 10\%, 20\%, \ldots, 100\%$ of $|E|$,  the precision values \mbox{$P_t$ (in \%) are $93, 87, 74, 68, 63, 59, 56, 53, 50, 47$, respectively.}
Thanks to sample edges in part  $\mathbf{F}$ of EdgeSketch, $P_t = 100\%$ for $t \leq 6\% |E|$.

\section{Conclusions} 
\label{sec:conclusions}
We introduced EdgeSketch, a compact graph representation that enables efficient analysis of massive graph streams. As the number of edges grows, EdgeSketch offers increasingly significant memory savings over lossless graph representations while providing runtime benefits and maintaining controlled accuracy. Furthermore, EdgeSketch operates in a fully streaming mode: after processing the stream of edges once and creating the sketch, the offline analysis relies solely on the information stored in the sketch. 

Our theoretical analysis shows that EdgeSketch provides unbiased estimators for key graph properties (e.g., node degrees, graph density, ratio of internal edges) with variance decreasing as $\mathcal{O}(1/m)$, where $m$ is the sketch size parameter.
 Moreover, by supporting analogues of set-theoretic operations, EdgeSketch enables a variety of graph algorithms to be implemented directly on the sketch.

We evaluated EdgeSketch on two important applications: community detection and graph reconstruction. Experiments show that EdgeSketch achieves substantial memory savings  and runtime improvement compared to adjacency list representations, while maintaining reliable accuracy. 
Moreover, in the graph reconstruction task, EdgeSketch achieves notably higher precision compared to prior sketch-based methods under equivalent memory constraints.



\newpage
\balance
\bibliographystyle{ACM-Reference-Format}
\bibliography{bibliography}

@article{holland1983sbm,
  title   = {Stochastic blockmodels: First steps},
  author  = {Holland, Paul W. and Laskey, Kathryn Blackmond and Leinhardt, Samuel},
  journal = {Social Networks},
  volume  = {5},
  number  = {2},
  pages   = {109--137},
  year    = {1983},
  doi     = {10.1016/0378-8733(83)90021-7}
}

@book{Penrose2003,
    author = {Penrose, Mathew},
    title = {Random Geometric Graphs},
    publisher = {Oxford University Press},
    year = {2003},
    month = {05},
    isbn = {9780198506263},
    doi = {10.1093/acprof:oso/9780198506263.001.0001},
    url = {https://doi.org/10.1093/acprof:oso/9780198506263.001.0001},
}

@inproceedings{nr,
     title={The Network Data Repository with Interactive Graph Analytics and Visualization},
     author={Ryan A. Rossi and Nesreen K. Ahmed},
     booktitle={AAAI},
     url={https://networkrepository.com},
     year={2015}
}

@article{Hamedani2021TrustRec,
  title={TrustRec: An Effective Approach to Exploit Implicit Trust and Distrust Relationships along with Explicit ones for Accurate Recommendations},
  author={Hamedani, M. and Ali, Irfan and Hong, Jiwon and Kim, Sang-Wook},
  journal={Computer Science and Information Systems},
  volume={18},
  number={1},
  pages={93--114},
  year={2021},
  doi={10.2298/CSIS200608039H},
  url={https://doi.org/10.2298/CSIS200608039H}
}

@article{UniProt,
    author = {The UniProt Consortium },
    title = {UniProt: the Universal Protein Knowledgebase in 2025},
    journal = {Nucleic Acids Research},
    volume = {53},
    number = {D1},
    pages = {D609-D617},
    year = {2024},
    month = {11},
    issn = {1362-4962},
    doi = {10.1093/nar/gkae1010},
    url = {https://doi.org/10.1093/nar/gkae1010},
    eprint = {https://academic.oup.com/nar/article-pdf/53/D1/D609/60719276/gkae1010.pdf},
}

@article{AdultBrain,
  title={Neuronal wiring diagram of an adult brain},
  author={Dorkenwald, Sven and Matsliah, Arie and Sterling, Amy R and Schlegel, Philipp and Yu, Szi-chieh and McKellar, Claire E and Lin, Albert and Costa, Marta and Eichler, Katharina and Yin, Yijie and others},
  journal={Nature},
  volume={634},
  number={8032},
  pages={124--138},
  year={2024},
  month={oct},
  publisher={Nature Publishing Group UK London},
  doi={10.1038/s41586-024-07558-y},
  url={https://doi.org/10.1038/s41586-024-07558-y}
}

@inproceedings{BoVWFI,
  author = "Paolo Boldi and Sebastiano Vigna",
  title = "The {W}eb{G}raph Framework {I}: {C}ompression Techniques",
  year = 2004,
  booktitle = "Proc. of the Thirteenth International World Wide Web Conference (WWW 2004)",
  address = "Manhattan, USA",
  pages = "595--601",
  publisher = "ACM Press"
}

@inproceedings{BRSLLP,
  author = "Paolo Boldi and Marco Rosa and Massimo Santini and Sebastiano Vigna",
  title = "Layered Label Propagation: A MultiResolution Coordinate-Free Ordering for Compressing Social Networks",
  booktitle = "Proceedings of the 20th international conference on World Wide Web",
  editor = "Sadagopan Srinivasan and Krithi Ramamritham and Arun Kumar and M. P. Ravindra and Elisa Bertino and Ravi Kumar",
  publisher = "ACM Press",
  year = 2011,
  pages = "587--596"
}

@InProceedings{YahooMusic,
  title = 	 {The Yahoo! Music Dataset and KDD-Cup’11},
  author = 	 {Dror, Gideon and Koenigstein, Noam and Koren, Yehuda and Weimer, Markus},
  booktitle = 	 {Proceedings of KDD Cup 2011},
  pages = 	 {3--18},
  year = 	 {2012},
  editor = 	 {Dror, Gideon and Koren, Yehuda and Weimer, Markus},
  volume = 	 {18},
  series = 	 {Proceedings of Machine Learning Research},
  month = 	 {21 Aug},
  publisher =    {PMLR},
  url = 	 {https://proceedings.mlr.press/v18/dror12a.html}
}

@article{gSketch,
author = {Zhao, Peixiang and Aggarwal, Charu C. and Wang, Min},
title = {gSketch: on query estimation in graph streams},
year = {2011},
issue_date = {November 2011},
publisher = {VLDB Endowment},
volume = {5},
number = {3},
issn = {2150-8097},
url = {https://doi.org/10.14778/2078331.2078335},
doi = {10.14778/2078331.2078335},
journal = {Proc. VLDB Endow.},
month = nov,
pages = {193–204},
numpages = {12}
}

@inproceedings{GSS,
author = {Tang, Nan and Chen, Qing and Mitra, Prasenjit},
title = {Graph Stream Summarization: From Big Bang to Big Crunch},
year = {2016},
isbn = {9781450335317},
publisher = {Association for Computing Machinery},
address = {New York, NY, USA},
url = {https://doi.org/10.1145/2882903.2915223},
doi = {10.1145/2882903.2915223},
booktitle = {Proceedings of the 2016 International Conference on Management of Data},
pages = {1481–1496},
numpages = {16},
location = {San Francisco, California, USA},
series = {SIGMOD '16}
}

@INPROCEEDINGS{Scube,
  author={Chen, Ming and Zhou, Renxiang and Chen, Hanhua and Jin, Hai},
  booktitle={2022 IEEE 42nd International Conference on Distributed Computing Systems (ICDCS)}, 
  title={Scube: Efficient Summarization for Skewed Graph Streams}, 
  year={2022},
  volume={},
  number={},
  pages={100-110},
  doi={10.1109/ICDCS54860.2022.00019}
  }

@article{Auxo,
author = {Jiang, Zhiguo and Chen, Hanhua and Jin, Hai},
title = {Auxo: A Scalable and Efficient Graph Stream Summarization Structure},
year = {2023},
issue_date = {February 2023},
publisher = {VLDB Endowment},
volume = {16},
number = {6},
issn = {2150-8097},
url = {https://doi.org/10.14778/3583140.3583154},
doi = {10.14778/3583140.3583154},
journal = {Proc. VLDB Endow.},
month = feb,
pages = {1386–1398},
numpages = {13}
}

@inproceedings{DeepWalk,
author = {Perozzi, Bryan and Al-Rfou, Rami and Skiena, Steven},
title = {DeepWalk: online learning of social representations},
year = {2014},
isbn = {9781450329569},
publisher = {Association for Computing Machinery},
address = {New York, NY, USA},
url = {https://doi.org/10.1145/2623330.2623732},
doi = {10.1145/2623330.2623732},
booktitle = {Proceedings of the 20th ACM SIGKDD International Conference on Knowledge Discovery and Data Mining},
pages = {701–710},
numpages = {10},
location = {New York, New York, USA},
series = {KDD '14}
}

@inproceedings{node2vec,
author = {Grover, Aditya and Leskovec, Jure},
title = {node2vec: Scalable Feature Learning for Networks},
year = {2016},
isbn = {9781450342322},
publisher = {Association for Computing Machinery},
address = {New York, NY, USA},
url = {https://doi.org/10.1145/2939672.2939754},
doi = {10.1145/2939672.2939754},
booktitle = {Proceedings of the 22nd ACM SIGKDD International Conference on Knowledge Discovery and Data Mining},
pages = {855–864},
numpages = {10},
location = {San Francisco, California, USA},
series = {KDD '16}
}

@inproceedings{LINE,
author = {Tang, Jian and Qu, Meng and Wang, Mingzhe and Zhang, Ming and Yan, Jun and Mei, Qiaozhu},
title = {LINE: Large-scale Information Network Embedding},
year = {2015},
isbn = {9781450334693},
publisher = {International World Wide Web Conferences Steering Committee},
address = {Republic and Canton of Geneva, CHE},
url = {https://doi.org/10.1145/2736277.2741093},
doi = {10.1145/2736277.2741093},
booktitle = {Proceedings of the 24th International Conference on World Wide Web},
pages = {1067–1077},
numpages = {11},
location = {Florence, Italy},
series = {WWW '15}
}

@inproceedings{GraRep,
author = {Cao, Shaosheng and Lu, Wei and Xu, Qiongkai},
title = {GraRep: Learning Graph Representations with Global Structural Information},
year = {2015},
isbn = {9781450337946},
publisher = {Association for Computing Machinery},
address = {New York, NY, USA},
url = {https://doi.org/10.1145/2806416.2806512},
doi = {10.1145/2806416.2806512},
booktitle = {Proceedings of the 24th ACM International on Conference on Information and Knowledge Management},
pages = {891–900},
numpages = {10},
location = {Melbourne, Australia},
series = {CIKM '15}
}

@inproceedings{NetMF,
author = {Qiu, Jiezhong and Dong, Yuxiao and Ma, Hao and Li, Jian and Wang, Kuansan and Tang, Jie},
title = {Network Embedding as Matrix Factorization: Unifying DeepWalk, LINE, PTE, and node2vec},
year = {2018},
isbn = {9781450355810},
publisher = {Association for Computing Machinery},
address = {New York, NY, USA},
url = {https://doi.org/10.1145/3159652.3159706},
doi = {10.1145/3159652.3159706},
booktitle = {Proceedings of the Eleventh ACM International Conference on Web Search and Data Mining},
pages = {459–467},
numpages = {9},
location = {Marina Del Rey, CA, USA},
series = {WSDM '18}
}

@inproceedings{VERSE,
author = {Tsitsulin, Anton and Mottin, Davide and Karras, Panagiotis and M\"{u}ller, Emmanuel},
title = {VERSE: Versatile Graph Embeddings from Similarity Measures},
year = {2018},
isbn = {9781450356398},
publisher = {International World Wide Web Conferences Steering Committee},
address = {Republic and Canton of Geneva, CHE},
url = {https://doi.org/10.1145/3178876.3186120},
doi = {10.1145/3178876.3186120},
booktitle = {Proceedings of the 2018 World Wide Web Conference},
pages = {539–548},
numpages = {10},
location = {Lyon, France},
series = {WWW '18}
}

@book{Klenke2014ProbabilityTheory,
  author    = {Achim Klenke},
  title     = {Probability Theory: A Comprehensive Course},
  edition   = {2},
  publisher = {Springer},
  year      = {2014}
}

@article{FastGumbelMax,
  author={Zhang, Yuanming and Wang, Pinghui and Qi, Yiyan and Cheng, Kuankuan and Zhao, Junzhou and Tian, Guangjian and Guan, Xiaohong},
  journal={IEEE Transactions on Knowledge and Data Engineering}, 
  title={Fast Gumbel-Max Sketch and its Applications}, 
  year={2023},
  volume={35},
  number={9},
  pages={9350-9363},
  doi={10.1109/TKDE.2023.3237857}}

@article{TemporalGraph,
author = {C. Coquidé, R. Cazabet},
title = {A Temporal Graph Dataset of Bitcoin Entity-Entity Transactions},
year = {2025},
publisher = {Nature Research},
volume = {12},
number = {337},
issn = {2150-8097},
url = {https://doi.org/10.1038/s41597-025-04595-8},
doi = {10.1038/s41597-025-04595-8},
journal = {Scientific Data},
}

@INPROCEEDINGS{FastAndAccurate,
  author={Gou, Xiangyang and Zou, Lei and Zhao, Chenxingyu and Yang, Tong},
  booktitle={2019 IEEE 35th International Conference on Data Engineering (ICDE)}, 
  title={Fast and Accurate Graph Stream Summarization}, 
  year={2019},
  volume={},
  number={},
  pages={1118-1129},
  doi={10.1109/ICDE.2019.00103}
  }

@misc{law_datasets,
  author       = {Paolo Boldi and Sebastiano Vigna},
  title        = {{Laboratory for Web Algorithmics (LAW) Datasets}},
  howpublished = {\url{https://law.di.unimi.it/datasets.php}},
  year         = {2026},
  note         = {Accessed: 2026-01-26},
  organization = {University of Milan}
}

@article{StreamingGraphEmbeddings,
  author={Yang, Dingqi and Qu, Bingqing and Yang, Jie and Wang, Liang and Cudre-Mauroux, Philippe},
  journal={IEEE Transactions on Knowledge and Data Engineering}, 
  title={Streaming Graph Embeddings via Incremental Neighborhood Sketching}, 
  year={2023},
  volume={35},
  number={5},
  pages={5296-5310},
  doi={10.1109/TKDE.2022.3149999}}

@inproceedings{NodeSketch,
author = {Yang, Dingqi and Rosso, Paolo and Li, Bin and Cudre-Mauroux, Philippe},
title = {NodeSketch: Highly-Efficient Graph Embeddings via Recursive Sketching},
year = {2019},
isbn = {9781450362016},
publisher = {Association for Computing Machinery},
address = {New York, NY, USA},
url = {https://doi.org/10.1145/3292500.3330951},
doi = {10.1145/3292500.3330951},
booktitle = {Proceedings of the 25th ACM SIGKDD International Conference on Knowledge Discovery \& Data Mining},
pages = {1162–1172},
numpages = {11},
location = {Anchorage, AK, USA},
series = {KDD '19}
}

@article{Lemiesz2023,
author = {Lemiesz, Jakub},
title = {Efficient Framework for Operating on Data Sketches},
year = {2023},
issue_date = {April 2023},
publisher = {VLDB Endowment},
volume = {16},
number = {8},
issn = {2150-8097},
url = {https://doi.org/10.14778/3594512.3594526},
doi = {10.14778/3594512.3594526},
journal = {Proc. VLDB Endow.},
month = {apr},
pages = {1967–1978},
numpages = {12}
}

@book{knuth97v2,
  address = {Boston},
  author = {Knuth, Donald E.},
  edition = {Third},
  isbn = {0201896842 9780201896848},
  publisher = {Addison-Wesley},
  refid = {174763889},
  title = {The Art of Computer Programming, Volume 2: Seminumerical Algorithms},
  year = 1997
}

@book{vanderVaart1998,
  author = {A. W. van der Vaart},
  title = {Asymptotic Statistics},
  publisher = {Cambridge University Press},
  year = {1998},
  series = {Cambridge Series in Statistical and Probabilistic Mathematics},
  isbn = {978-0-521-78450-4}
}

@article{Lemiesz21,
	author    = {Jakub Lemiesz},
	title     = {On the algebra of data sketches},
	journal   = {Proc. {VLDB} Endow.},
	volume    = {14},
	number    = {9},
	pages     = {1655--1667},
	year      = {2021},
	url       = {http://www.vldb.org/pvldb/vol14/p1655-lemiesz.pdf},
	doi       = {10.14778/3461535.3461553},
	bibsource = {dblp computer science bibliography, https://dblp.org}
}

@Book{Devr86,
  Title                    = {Non-Uniform Random Variate Generation},
  Author                   = {Luc Devroye},
  Publisher                = {Springer-Verlag},
  Year                     = {1986},

  Address                  = {New York, NY, USA}
}

@article{ModularityWeighted,
  author = {M. E. J. Newman},
  doi = {10.1103/PhysRevE.70.056131},
  journal = {Phys. Rev. E},
  month = nov,
  number = 5,
  numpages = {9},
  pages = 056131,
  publisher = {American Physical Society},
  title = {Analysis of weighted networks},
  volume = 70,
  year = 2004
}

@book{Barabasi2016,
  address = {Cambridge},
  author = {Barabási, Albert-László and Pósfai, Márton},
  isbn = {9781107076266 1107076269},
  publisher = {Cambridge University Press},
  title = {Network science},
  year = 2016
}

@article{newman03fast,
  author = {Newman, M.E.J.},
  journal = {Physical Review E},
  month = {September},
  title = {Fast algorithm for detecting community structure in networks},
  volume = 69,
  year = 2003
}

@article{Leiden,
  title={From Louvain to Leiden: guaranteeing well-connected communities},
  author={Vincent Antonio Traag and Ludo Waltman and Nees Jan van Eck},
  journal={Scientific Reports},
  year={2018},
  volume={9},
 }

@InProceedings{Fiona_2024,
  author =	{Louf, Baptiste and McDiarmid, Colin and Skerman, Fiona},
  title =	{{Modularity and Graph Expansion}},
  booktitle =	{15th Innovations in Theoretical Computer Science Conference (ITCS 2024)},
  pages =	{78:1--78:21},
  series =	{Leibniz International Proceedings in Informatics (LIPIcs)},
  ISBN =	{978-3-95977-309-6},
  ISSN =	{1868-8969},
  year =	{2024},
  volume =	{287},
  editor =	{Guruswami, Venkatesan},
  publisher =	{Schloss Dagstuhl -- Leibniz-Zentrum f{\"u}r Informatik},
  address =	{Dagstuhl, Germany},
  doi =		{10.4230/LIPIcs.ITCS.2024.78},
 }

@article{Greedy2024,
author = {Lanciano, Tommaso and Miyauchi, Atsushi and Fazzone, Adriano and Bonchi, Francesco},
title = {A Survey on the Densest Subgraph Problem and its Variants},
year = {2024},
issue_date = {August 2024},
publisher = {Association for Computing Machinery},
address = {New York, NY, USA},
volume = {56},
number = {8},
issn = {0360-0300},
url = {https://doi.org/10.1145/3653298},
doi = {10.1145/3653298},
journal = {ACM Comput. Surv.},
month = apr,
articleno = {208},
numpages = {40}
}

@article{karger1996,
author = {Karger, David R. and Stein, Clifford},
title = {A new approach to the minimum cut problem},
year = {1996},
issue_date = {July 1996},
publisher = {Association for Computing Machinery},
address = {New York, NY, USA},
volume = {43},
number = {4},
issn = {0004-5411},
url = {https://doi.org/10.1145/234533.234534},
doi = {10.1145/234533.234534},
journal = {J. ACM},
month = jul,
pages = {601–640},
numpages = {40}
}

@article{KNUTH_1978,
title = {The expected linearity of a simple equivalence algorithm},
journal = {Theoretical Computer Science},
volume = {6},
number = {3},
pages = {281-315},
year = {1978},
issn = {0304-3975},
doi = {https://doi.org/10.1016/0304-3975(78)90009-9},
url = {https://www.sciencedirect.com/science/article/pii/0304397578900099},
author = {Donald E. Knuth and Arnold Schönhage}
}

@INPROCEEDINGS {Thota_2021,
author = { Thota, Saigopal and Jain, Mridul and Kamat, Nishad and Malikireddy, Saikiran and Eranti, Pruthvi Raj and Kuruvilla, Albin },
booktitle = { 2021 IEEE International Conference on Big Data (Big Data) },
title = {{ Building Graphs at a Large Scale: Union Find Shuffle }},
year = {2021},
volume = {},
ISSN = {},
pages = {4146-4152},
doi = {10.1109/BigData52589.2021.9671575},
url = {https://doi.ieeecomputersociety.org/10.1109/BigData52589.2021.9671575},
publisher = {IEEE Computer Society},
address = {Los Alamitos, CA, USA},
month =Dec}

\end{document}